\documentclass[a4paper,UKenglish,cleveref, autoref, thm-restate]{lipics-v2021}

\hideLIPIcs  

\usepackage{cite}
\usepackage{amsmath,amssymb,amsfonts}
\usepackage{algorithmic}
\usepackage{graphicx}
\usepackage{textcomp}
\usepackage[dvipsnames]{xcolor}
\usepackage{amsthm, tikz, hyperref}

\usepackage{todonotes}

\newcommand{\nat} {\mathbb{N}} 
\newcommand{\rat} {\mathbb{Q}} 
\newcommand{\alg} {\overline{\mathbb{Q}}}
\newcommand{\Acc} {\operatorname{Acc}}

\title{On Expansions of Monadic Second-Order Logic with Dynamical Predicates} 

\author{Joris {Nieuwveld}}{Max Planck Institute for Software Systems, Saarland Informatics Campus, Germany}{jnieuwve@mpi-sws.org}{https://orcid.org/0009-0002-0339-1230}{}

\author{Jo\"el Ouaknine}%
   {Max Planck Institute for Software Systems, Saarland Informatics
     Campus, Germany}{joel@mpi-sws.org}{https://orcid.org/0000-0003-0031-9356}{Also affiliated with Keble College, Oxford as
\href{https://www.emmy.network/}{emmy.network} Fellow, and supported by ERC grant DynAMiCs (101167561) and DFG grant
389792660 as part of \href{https://perspicuous-computing.science}{TRR 248}.}

\authorrunning{J. Nieuwveld and J. Ouaknine} 

\Copyright{Joris {Nieuwveld} and Jo\"el {Ouaknine}} 



\category{} 

\relatedversion{} 

\nolinenumbers 

\EventEditors{Pawe\l{} Gawrychowski, Filip Mazowiecki, and Micha\l{} Skrzypczak}
\EventNoEds{3}
\EventLongTitle{50th International Symposium on Mathematical Foundations of Computer Science (MFCS 2025)}
\EventShortTitle{MFCS 2025}
\EventAcronym{MFCS}
\EventYear{2025}
\EventDate{August 25--29, 2025}
\EventLocation{Warsaw, Poland}
\EventLogo{}
\SeriesVolume{345}
\ArticleNo{83}

\begin{document}

\maketitle

\begin{abstract}
Expansions of the monadic second-order (MSO) theory of the structure
  $\langle \nat ; < \rangle$ have been a fertile and active area of
  research ever since the publication of the seminal papers of B\"uchi
  and Elgot \& Rabin on the subject in the 1960s. In the present paper, we
  establish decidability of the MSO theory of
  $\langle \nat ; <,P \rangle$, where $P$ ranges over a large class of unary
  ``dynamical'' predicates, i.e., sets of non-negative
  values assumed by certain integer linear recurrence sequences. One
  of our key technical tools is the novel concept of \emph{(effective) prodisjunctivity}, which we expect may also find independent applications
  further afield.
\end{abstract}

\section{Introduction}
\label{sec : intro}

The monadic second-order (MSO) theory of the structure
$\langle \nat ; < \rangle$ has been a foundational pillar of the field
of automated verification, and more generally the area of logic in
computer science, for many decades.
Arguably the most important paper on the topic is due to
B\"uchi~\cite{buchi1966symposium}, who in the early 1960s established
decidability of this theory through a deep and fecund connection between logic
and automata theory.

Shortly thereafter, in yet another seminal piece of
work~\cite{elgot1966decidability}, Elgot and Rabin devised the
\emph{contraction method} to establish decidability of expansions of
this base theory by various ``arithmetic'' unary predicates
$P \subseteq \nat$.\footnote{In this paper, we adopt the convention
  that the set $\nat$ of natural numbers contains $0$. We also use the
  adjective ``positive'' with the meaning of ``non-negative''.}  In
particular, they proved decidability of the MSO theory of
$\langle \nat ; <,P \rangle$, where $P$ could for example be taken to
be the set $2^\nat$ of powers of $2$, or the set $\mathsf{Sq}$ of
perfect squares, or the set $\mathsf{Fac}$ of factorial numbers, and
so on.

Much progress followed in the ensuing decades, notably thanks to
Sem\"enov~\cite{semenov1984logical}, who introduced the notion of
\emph{effective almost periodicity}, and Carton
and Thomas~\cite{cartonthomas}, who substantially refined Elgot
and Rabin's contraction method into the notion of 
\emph{effective profinite ultimate periodicity}.
Other notable works in this area include articles by
Rabinovich~\cite{rabinovich}, Rabinovich and
Thomas~\cite{rabinovich2006decidable}, and Berth\'e \emph{et al.}~\cite{berthe2024decidability}.

In the present paper, we significantly extend this line of research by
considering a large class of ``dynamical'' predicates derived from
integer linear recurrence sequences (LRS)\@.\footnote{The ``dynamical''
  moniker was chosen to reflect the close kinship of these predicates
  with discrete-time linear dynamical systems; see, e.g.,
  \cite{Baier0JKLLOPW021,DCostaKMOSW22,KarimovLOPVWW22,KarimovKO022,LucaOW22,KarimovKNO023,AghamovBKOP24}.}  The complexity of an LRS can be
measured in various ways; chief among them are the number of distinct
\emph{dominant characteristic roots} of the LRS, and whether the LRS is \emph{simple} or
not.\footnote{These notions are properly defined in Sec.~\ref{sec : LRS}.}
For $P$ the set of positive values of an LRS having a single
dominant root, such as the set of Fibonacci numbers, the decidability of
the MSO theory of the structure $\langle \nat ; <,P \rangle$ is
readily established via Elgot and Rabin's contraction
method;\footnote{LRS having a single positive dominant root are particularly well
  behaved: they are either constant, or effectively ultimately
  monotonically increasing (or decreasing),
  effectively procyclic, and effectively sparse
  (see~\cite{berthe2024decidability} for precise definitions of these
  notions and a proper account of these facts). In the case of a
  negative dominant root, essentially the same behaviour is observed
  by restricting to the positive terms.} see
also~\cite{berthe2024decidability} in which expansions of $\langle
\nat ; < \rangle$ by adjoining multiple predicates obtained from such
LRS are thoroughly investigated. Unfortunately, virtually nothing is
known when considering LRS having more than a single dominant characteristic root.

Our main contribution is the following result:

\begin{theorem}\label{thm : MSO decidable}
    Let $P = \{u_n \colon n \in \nat\} \cap \nat$ for $\langle u_n
    \rangle_{n=0}^\infty$ a 
    non-degenerate, simple, integer-valued linear recurrence
    sequence with two dominant roots. 
    Then the MSO theory of the structure $\langle \nat ; < , P\rangle$ is decidable.
\end{theorem}

An example of an LRS satisfying the hypotheses of Thm.~\ref{thm :
  MSO decidable} is the sequence $\langle u_n\rangle_{n=0}^\infty$ defined by
\begin{equation}\label{eq : example LRS}
    u_{n+3} = 6u_{n+2} - 13u_{n+1} + 10 u_n
\end{equation}
and $u_0 = 2, u_1 = 4$, and $u_2 = 7$. 
Its subsequent values are:
$10$,
$9$,
 $-6$,
 $-53$,
 $-150$,
 $-271$,
 $-206$,
 $787$,
 $4690$,
 $15849$,
 $41994$,
 $92827$,
 $169530$,
 $230369$,
 $106594$, \dots.
One can readily verify that $\langle u_n \rangle_{n=0}^\infty$ satisfies the formula
\begin{equation}\label{eq : example LRS poly exp form}
    u_n = \frac{1}{2}(2+i)^n + \frac{1}{2}(2-i)^n + 2^n \, ,
\end{equation}
and that $u_n$ is both infinitely often positive and infinitely often negative.
Moreover, although $\lim_{n\to \infty}|u_n| = +\infty$, $|u_n|$ is \emph{not}
monotonically increasing: for all $N \in \nat$ there is an $n \in \nat$ such that
$|u_{n+N}| < |u_n|$.

For our LRS $\langle u_n\rangle_{n=0}^\infty$, let $P$ be as per Thm.~\ref{thm : MSO decidable}. Viewing $P$ as an
ordered set, we have:
\begin{equation*}
    P = \{ u_0, u_1, u_2, u_4, u_3, u_{10}, u_{11}, u_{12}, u_{13},
        u_{14}, u_{17}, u_{15}, u_{16}, u_{24},  u_{25}, u_{26}, u_{27}, u_{28},   u_{30},\dots\}.
\end{equation*}

One immediately notices that there are two complications at play
here. The first one is that we are ``throwing away'' all the negative
values of our LRS (this is of course necessary since we are working over
domain $\nat$ rather than $\mathbb{Z}$). This restriction is, however,
entirely benign in view of the following result:

\begin{corollary}\label{cor : Z-MSO decidable}
    Let $\overline{P} = \{u_n \colon n \in \nat\}$ for $\langle u_n
    \rangle_{n=0}^\infty$ a 
    non-degenerate, simple, integer-valued linear recurrence
    sequence with two dominant roots. 
    Then the MSO theory of the structure $\langle \mathbb{Z} ; 0, < , \overline{P} \rangle$ is decidable.
\end{corollary}

This result is straightforwardly obtained from Thm.~\ref{thm : MSO
  decidable} via an application of Shelah's celebrated ``composition method'' in model
theory~\cite{She75}. In the case at hand, one can directly invoke, for
example, \cite[Cor.~6]{Tho97}, since the structure
$\langle \mathbb{Z} ; 0, < , \overline{P} \rangle$ is isomorphic to
the ordered sum $\langle \mathbb{Z}-\nat ; < , P^{-}\rangle +
\langle \nat ; < , P\rangle$, where the second summand is as per
Thm.~\ref{thm : MSO decidable}, and the first structure is obtained
by setting $P^{-} = \{u_n \colon n \in \nat\} \cap
(\mathbb{Z}-\nat)$. Intuitively speaking, MSO sentences over
$\langle \mathbb{Z} ; < , \overline{P} \rangle$ can be faithfully decomposed into
component subformulas dealing exclusively with either positive or
negative values of our LRS, with the truth value of the original
sentence obtained by appropriately piecing together truth values of each of
the sub-sentences within their respective structures.

The second complication is that, as noted earlier, the ordering of the positive values of
our LRS does not respect the index ordering of the LRS\@. This
ostensibly precludes the direct application of classical techniques such as
Elgot and Rabin's contraction method (or more generally Carton and
Thomas's effective profinite ultimate periodicity criterion), or Sem\"enov's
toolbox of effective almost periodicity.

In order to prove Thm.~\ref{thm : MSO decidable}, we therefore
rely instead on a new concept, that of \emph{(effective) prodisjunctivity}, which we introduce shortly. 
Prodisjunctivity is itself predicated on the notion of disjunctivity,
whose application to the decidability of logical theories
(under the alternative terminology of ``weak normality'') was
pioneered by Berth\'e
\emph{et al.}~\cite{berthe2024decidability}.
In particular,
Berth\'e \emph{et al.}\ study the MSO theory of the structure
$\langle \nat ; <, 2^\nat, \mathsf{Sq} \rangle$, and establish
decidability assuming that the binary expansion of $\sqrt{2}$ is
disjunctive, i.e., contains every finite bit pattern as a factor
infinitely often.

The hypothesis that $\sqrt{2}$ is disjunctive in base $2$ is widely
expected by number theorists to be true, but remains a major unsolved
problem. Irrational algebraic numbers, along with most known
transcendental numbers such as $e$ and $\pi$, are in fact believed to
satisfy a stronger property, that of being \emph{normal} in every
integer base: a real number $\alpha$ is normal in base $b$ provided
that every finite word $w \in \{0,\dots,b-1\}^*$ appears with
frequency $b^{-|w|}$ in the base-$b$ expansion of $\alpha$. And
although Borel~\cite{emile1909probabilites} showed over a century ago
that non-normal real numbers have null Lebesgue measure, establishing
normality (or even disjunctivity) of everyday irrational numbers has
remained fiercely elusive;
see~\cite{Har02,queffelec2006old,Kui12,Bug12} for more detailed
accounts of results, research, and open problems in this area.

Given a finite alphabet $\Sigma$ together with a subset
$S \subseteq \Sigma$, we say that an infinite sequence in
$\Sigma^{\omega}$ is \emph{disjunctive relative to $S$} if every
finite word over $S$ appears infinitely often as a factor in the
sequence. We are now in a position to define (effective)
prodisjunctivity:




\begin{definition}
\label{def : pronormal}
  Let $\langle p_m \rangle_{m=0}^{\infty}$ be an infinite integer-valued
  sequence. For any integer $M \geq 2$, let $S_M$ be the set of
  residue classes modulo $M$ that appear infinitely often in
  $\langle p_m \rangle_{m=0}^{\infty}$:
\begin{equation}\label{eq : definition S}
    S_M = \big\{s \in \{0,\dots, \, M-1\} \colon \exists^\infty m \in \nat
    \colon p_m \equiv s \!\!\! \pmod M \big\} \, .
\end{equation}

We say that the sequence $\langle p_m \rangle_{n=0}^{\infty}$ is
\textbf{prodisjunctive} if, for all $M \geq 2$, the sequence of
residues $\langle p_m\bmod M \rangle_{m=0}^{\infty}$ is
disjunctive relative to $S_M$.

For \textbf{effectiveness}, we require in addition that the sequence
$\langle p_m\bmod M \rangle_{m=0}^{\infty}$
be computable and that, for each $M$, the set $S_M$,
together with the smallest index threshold $N_M$ beyond which
all residue classes of the sequence lie in $S_M$ (i.e., for all
$m \geq N_M$, $(p_m\bmod M) \in S_M$), also be computable.
\end{definition}

Let us now sketch how prodisjunctivity relates to Thm.~\ref{thm : MSO decidable}. 
Recall from \eqref{eq : example LRS} our LRS $\langle u_n\rangle_{n=0}^\infty$, along with its infinite set $P$ of positive values (as per the notation of Thm.~\ref{thm : MSO decidable}),
and let $\langle p_m\rangle_{m=0}^\infty$ be a strictly increasing
enumeration of $P$; in other words, $p_{m-1}$ is the $m$th smallest element in $P$.
We will establish the following instrumental result:

\begin{theorem}\label{thm : intro weak normality}
  Let $\langle p_m\rangle_{m=0}^\infty$ be as above, i.e., the
  strictly increasing sequence of positive terms of some non-degenerate, simple,
integer-valued LRS having two dominant roots. Then
$\langle p_m\rangle_{m=0}^\infty$ is effectively prodisjunctive.
\end{theorem}

Let us return to our example and set $M = 5$. 
Then one easily shows that
\begin{equation}\label{eq : example intro modular behavior}
    u_n\bmod 5 = \begin{cases}
        0 & \text{if } n\equiv 3 \pmod 4 \\
        2 & \text{if } n = 0, \text{ or if } n \equiv 2 \pmod 4  \\
        4 & \text{otherwise} \, .
    \end{cases}
\end{equation}
It follows that $S_5 = \{0,2,4\}$, $N_5 = 0$, and
\begin{equation*}
  \langle p_m\bmod 5\rangle_{m=0}^\infty =
    \langle 2, 4, 2, 4, 0, 2, 0, 4, 4, 2, 4, 0, 4, 4, 4,2,
    0, 4, 2, 4, 2, 0, 4, 4, 4, 2, 0, 0, 4, 4, 2, \dots\rangle
\end{equation*}
is in $\{0,2,4\}^\omega$.

Computing the first million terms of $\langle p_m\bmod 5\rangle_{m=0}^\infty$, one does not encounter the factor $\langle 0,0,0 \rangle$ once within that initial segment. 
Nevertheless, according to Thm.~\ref{thm : intro weak normality}, it should appear infinitely often! 
Indeed, we prove this in Sec.~\ref{sec : example}, and moreover
construct an index of roughly $2.18 \cdot 10^{59}$ where this
occurs. A more involved calculation (not shown here) can however establish that the very first
occurrence of $\langle 0,0,0 \rangle$ within $\langle p_m \bmod
5\rangle_{m=0}^\infty$ is at index $8479226$.

The above discussion suggests that, whilst the sequence
$\langle p_m\rangle_{m=0}^\infty$ is provably disjunctive, it is
by all accounts not normal, i.e., a given factor $w \in \{0,2,4\}^*$ does
not necessarily appear with frequency $3^{-|w|}$.  This is however
unsurprising in view of~\eqref{eq : example intro modular behavior}:
the residue class $4$ should statistically appear approximately twice
as often as either of the other two residue classes, and indeed it is
possible to prove that this is asymptotically the case.

Theorem~\ref{thm : intro weak normality} is the key technical device
underpinning our main result, Thm.~\ref{thm : MSO
  decidable}. One of the critical mathematical ingredients entering its
proof is Baker's theorem on linear forms in logarithms of algebraic
numbers. We also make use of various automata-theoretic, topological,
and combinatorial constructions and tools.

Let us now briefly comment on the motivation and scope of
Thm.~\ref{thm : MSO decidable}. As noted earlier, linear recurrence
sequences and linear dynamical systems are intimately related; in
particular, decision procedures for the former often underpin
verification algorithms for the latter. Write $P_{\boldsymbol{u}}$ to
denote the set of positive values of a given integer LRS $\boldsymbol{u}$.
We already pointed out that Elgot and Rabin's contraction method
enables one to establish decidability of the MSO theory of
$\langle \nat ; <,  P_{\boldsymbol{u}}\rangle$ provided that
$\boldsymbol{u}$ has a single dominant root, and that
\cite[Thm.~5.1]{berthe2024decidability} provides general conditions
under which the MSO theory of
$\langle \nat ; <,  P_{\boldsymbol{u}^{(1)}},\ldots, P_{\boldsymbol{u}^{(d)}}\rangle$
is decidable, once again under the assumption that each LRS
$\boldsymbol{u}^{(i)}$ possesses a single dominant root.
On the other hand, the situation concerning predicates derived from LRS having more than
one dominant  root is entirely uncharted. The present paper can thus
be seen as taking the first significant step (in some six decades!\@)
in this direction.

Nevertheless, the hypotheses of Thm.~\ref{thm : MSO decidable} (the
LRS must be non-degenerate, simple, and have exactly two dominant
roots) may at first glance appear quite restrictive. It is worth noting, however, that
if one were to pick LRS ``at random'', then under most ``reasonable''
probability distributions, with probability tending to $1$ such LRS
would have either a single dominant root, or be non-degenerate,
simple, and have exactly two dominant roots. In other words, the vast
majority of LRS are of one of these two types, and thus fall under the
scope of the Elgot-Rabin contraction method or our own Thm.~\ref{thm
  : MSO decidable}. (This folklore result follows from the fact that a
random polynomial with real coefficients will almost surely have
either a single dominant real root, or two dominant complex conjugate
roots; and that non-degeneracy and simplicity are null-measure
conditions. For a more in-depth discussion of these matters, see the work of Dubickas and
Sha~\cite{dubickas2016positive} and citations therein.) In any event,
we briefly return in Sec.~\ref{sec : conclusion} to the question of whether the hypotheses
of Thm.~\ref{thm : MSO decidable} could be weakened in any way.

\section{Preliminaries}
\subsection{Automata Theory}
An \emph{alphabet} $\Sigma$ is a finite, non-empty set of letters and
$\Sigma^*$ and $\Sigma^\omega$ denote respectively
the sets of finite and infinite words over $\Sigma$. 
As words can be viewed as sequences, we freely identify the finite
word $w_0w_1 \cdots w_k$ with the finite sequence $\langle
w_0,w_1,\dots,w_k\rangle$, and the infinite word $w_0w_1 \cdots $ with
the infinite sequence $\langle w_0,w_1,\dots \rangle$, and conversely.
A finite word $w' = w'_0w'_1 \cdots w'_k$ is a \emph{factor} of an
(in)finite word $w$ if there is an index $n$ such that $\langle
w_n,\dots,w_{n+k} \rangle = \langle w'_0, \dots, w'_k \rangle$.
We let $|w|$ denote the length a finite word $w$. 

An infinite word $w$ is \emph{recursive} (or \emph{computable})
if one can compute $w_j$ for every $j \ge 0$, and
is \emph{disjunctive} if each $w' \in {\Sigma}^*$ appears infinitely often as a factor of $w$.
Other authors occasionally use the terminology \emph{weakly normal} or \emph{rich} instead of disjunctive.
An infinite word $w$ is \emph{periodic} if $w = v^\omega$ for some $v \in \Sigma^+$ and \emph{ultimately periodic} if $w = u v^\omega$ for some $u \in \Sigma^*$ and $v \in \Sigma^+$. 
Here, for the smallest possible such $u$ and $v$, $|u|$ and $|v|$ are
referred to as the \emph{preperiod} and \emph{period}, respectively.

A \emph{deterministic Muller automaton}
$\mathcal{A} = (\Sigma, Q, q_\text{init}, \delta, \mathcal{F})$
consists of an alphabet $\Sigma$, a finite set of states $Q$, an
initial state $q_\text{init} \in Q$, a transition function
$\delta \colon Q \times \Sigma \to Q$, and an accepting family of sets
$\mathcal{F} \subseteq 2^Q$.  The \emph{run} of
$w\in \Sigma^* \cup \Sigma^\omega$ on $\mathcal{A}$ is the sequence of
states visited while reading $w$ starting in $q_\text{init}$ and
repeatedly updating the state using the transition function.  We say
that $\mathcal{A}$ \emph{accepts} $w \in \Sigma^\omega$ if the set of
states visited infinitely often upon reading $w$ belongs to
$\mathcal{F}$.  The \emph{acceptance problem} of a recursive word
$w \in \Sigma^\omega$ is the question of determining, given a
deterministic Muller automaton $\mathcal{A}$ with alphabet $\Sigma$,
whether $\mathcal{A}$ accepts $w$.  We denote the acceptance
problem by $\Acc_w$.

Berth\'e \emph{et al.}\ established the following~\cite[Thm.~4.16]{berthe2024decidability}:
\begin{theorem}\label{thm : weakly normal decidable}
    If $w \in \Sigma^\omega$ is recursive and disjunctive, then $\Acc_w$ is decidable.
\end{theorem}

A \emph{deterministic finite transducer} $\mathcal{B} =
(\Sigma_{\text{in}}, \Sigma_{\text{out}}, Q, q_\text{init}, \delta)$ is
given by an input alphabet $\Sigma_\text{in}$, output alphabet
$\Sigma_\text{out}$, set of states $Q$, initial state $q_\text{init}
\in Q$, and transition function $\delta \colon Q \times \Sigma_\text{in}
\to Q \times \Sigma_\text{out}^*$.
The transducer $\mathcal{B}$ starts in state $q_\text{init}$.
It reads a word $w \in \Sigma_\text{in}^* \cup
\Sigma_\text{in}^\omega$ and upon reading letter $a$ whilst in state
$q$, it computes $(q', w') = \delta(q, a)$, moves to state $q'$, and
concatenates $w'$ to the output string.
Write $\mathcal{B}(w) \in \Sigma_\text{out}^* \cup
\Sigma_\text{out}^\omega$ to denote the output word thus computed upon
reading $w \in \Sigma_\text{in}^* \cup
\Sigma_\text{in}^\omega$.

We now recall a Turing-reducibility property between
word acceptance problems~\cite[Lem.~4.5]{berthe2024decidability}:
\begin{lemma}\label{thm : transducer reduction}
Let $w \in \Sigma^\omega$ and $\mathcal{B}$ be a deterministic finite transducer. Then the problem $\Acc_{\mathcal{B}(w)}$ reduces to $\Acc_{w}$.
\end{lemma} 

\subsection{Monadic Second-Order Logic}

A \emph{(unary or monadic) predicate} $P$ is a function
$P \colon \nat \to \{0, 1\}$, which equivalently we can interpret as a
subset $P \subseteq \nat$.  For $P \subseteq \nat$ an infinite
predicate, let us write $\langle p_m\rangle_{m=0}^\infty$ to denote
the enumeration of $P$ in increasing order, i.e., such that $p_{m-1}$ is the $m$th smallest element of $P$.
The \emph{characteristic word} $w \in \{0, 1\}^\omega$ of $P$ is
obtained by setting $w_n = P(n)$.  We then have:
\begin{equation}\label{eq : characteristic word P}
    w = 0^{p_0}10^{p_1 - p_0-1}10^{p_2 - p_1-1}10^{p_3 - p_2-1} \cdots \, .
\end{equation}
A predicate $P$ is \emph{sparse} if for every $N \ge 0$, there is $M \ge 0$ such that $p_{m+1} - p_m \ge N$ for all $m \ge M$.
The predicate is \emph{effectively sparse} if $M$ can always be
computed from $N$.

\emph{Monadic second-order logic (MSO)} is an extension of first-order logic
over signature $\{=, <, \in\}$
that allows quantification over subsets of the universe $\nat$.
We also consider expansions of MSO by various fixed unary predicates
$P \subseteq \nat$; in other words (abusing notation), the signature is expanded by a predicate
symbol $P$, with interpretation the given subset $P \subseteq \nat$.
We refer the reader to~\cite{mso-notes} for a thorough contemporary overview of MSO\@.

The deep connection between MSO and automata theory was brought to
light in the seminal work of B\"uchi; see, for example, \cite[Thms.~5.4 and 5.9]{thomas1997languages}.

\begin{theorem}\label{thm : Buchi MSO char word}
    The decidability of the MSO theory of the structure $\langle \nat ; < , P \rangle$ is Turing equivalent to $\Acc_w$, where $w$ is the characteristic word of $P$.
\end{theorem}

As noted earlier, Elgot and Rabin devised the \emph{contraction
  method} to establish decidability of the MSO theory of
$\langle \nat ; < , P \rangle$, for various ``arithmetic'' predicates
$P$. The following proposition is a variation on their method:

\begin{proposition}\label{prop : reduction one unary predicate} 
Let $P \subseteq \nat$ be an infinite, recursive, and effectively
sparse predicate with enumeration $\langle p_m\rangle_{m=0}^\infty$.
If, for each $M \ge 2$, $\Acc_{\langle p_m\bmod M\rangle_{m=0}^\infty}$ is decidable, the MSO theory of the structure $\langle \nat; <, P \rangle$ is decidable.
\end{proposition}
\begin{proof}
  By Thm.~\ref{thm : Buchi MSO char word}, it is sufficient to be able
  to decide whether a deterministic Muller automaton
  $\mathcal{A} = (\{0,1\}, Q, q_{\text{init}}, \delta, \mathcal{F})$
  accepts the characteristic word \eqref{eq : characteristic word P}.
  Let us restrict $\mathcal{A}$ to a directed graph $G$ with nodes $Q$
  and $0$-transitions as arrows.

  By construction, every node in $G$ has outdegree $1$ and so each
  state is in at most one cycle in $G$.  We can therefore compute the
  least common multiple of the cycle lengths in $G$ (call this number
  $M$) and the longest path to a cycle (call this number $N$).  Then,
  for all states $q$ and numbers $n \ge M + N$ and $d \ge 1$, reading
  $0^n$ and $0^{n+dM}$ leads to journeying through the exact same set
  of states and ending up in the same state.

  Let $K$ be such that $p_{m+1} - p_m \ge M + N+1$ for all $m \ge K$
  (which can be computed as $P$ is effectively sparse).  We construct
  a deterministic finite transducer $\mathcal{B}$ that hard-codes the
  initial segment
  $w_\text{init} := 0^{p_0}10^{p_1-p_0-1}1\cdots
  0^{p_{K+1}-p_{K}-1}1$.  For $m \ge K$, after reading
  $(p_m\bmod M)$ and $(p_{m+1}\bmod M)$, $\mathcal{B}$ outputs
  $0^{k_m}1$, where $M+N < k_m \le 2M + N$ is congruent to
  $p_{m+1} - p_m - 1$ modulo $M$.  Then, by construction, a state $q$
  is visited infinitely often upon reading the characteristic word of
  $P$ if and only if $q$ is visited infinitely often when
  $\mathcal{A}$ reads
  $w_\text{init}\mathcal{B}(\langle p_m\bmod M\rangle_{m=K+1}^\infty)$.

    Recall that $\Acc_{\langle p_m\bmod M\rangle_{m=0}^\infty}$ is assumed to be decidable. 
    Then $\Acc_{\langle p_m\bmod M)_{m=K+1}^\infty}$ is also
    decidable (by hard-coding the initial segment) and thus
    $\Acc_{\mathcal{B}(\langle p_m\bmod M\rangle_{m=K+1}^\infty)}$ is decidable by Thm.~\ref{thm : transducer reduction}.
    Therefore $\Acc_{w_\text{init}\mathcal{B}(\langle p_m\
      \mathrm{mod}\ M\rangle_{m=K+1}^\infty)}$ is decidable (by again
    hard-coding the initial segment), and by construction, $\mathcal{A}$ accepts
    the characteristic word of $P$ if and only if $\mathcal{A}$
    accepts $w_\text{init}\mathcal{B}(\langle p_m\bmod M\rangle_{m=K+1}^\infty)$. 
    Hence $\Acc_w$ is decidable, as required.
    %
    %
    %
\end{proof}

\subsection{Linear Recurrence Sequences}\label{sec : LRS}
A number $\alpha \in \mathbb{C}$ is \emph{algebraic} if $F(\alpha) = 0$ for some non-zero polynomial $F \in \mathbb{Z}[X]$.
The unique such polynomial (up to multiplication by a constant) of least degree is the \emph{minimal polynomial} of $\alpha$.
We write $\alg$ to denote the field of algebraic numbers.

A \emph{linear recurrence sequence} over a ring $R$ (an
\emph{$R$-LRS}) is a sequence
$\langle u_n\rangle_{n=0}^\infty \in R^\omega$ such that there are
numbers $c_1,\dots,c_d \in R$, with $c_d \ne 0$, having the property that, for all $n \in \nat$,
\begin{equation*}
    u_{n+d} = c_1u_{n+d-1} + \cdots + c_du_n \, .
\end{equation*}
The entire sequence is therefore completely specified by the $2d$
numbers $c_1,\dots,c_d$ and $u_0,\dots,u_{d-1}$. Although there may
be several such recurrence relations, we shall assume that we are
always dealing with the (unique) one for which $d$ is minimal.
In the remainder of the paper, whenever $R$ is not specified,
we are working over the ring of integers $\mathbb{Z}$.
The \emph{characteristic polynomial} of the LRS is given by
$F(X) = X^d - c_1X^{d-1} - \cdots - c_d \in R[X]$.  The
\emph{characteristic roots} of the LRS are the roots of its
characteristic polynomial, and the \emph{multiplicity} of a
characteristic root $\lambda$ is its multiplicity as a root of $F(X)$.
Every $\alg$-LRS $\langle u_n \rangle_{n=0}^\infty$ with
characteristic roots $\lambda_1,\dots,\lambda_K$ admits a unique
\emph{exponential-polynomial} representation given by
\begin{equation*}
    u_n = \sum_{k=1}^K Q_k(n)\lambda_k^n \, ,
\end{equation*}
where $Q_k \in \overline{\rat}[X]$ has degree the multiplicity of $\lambda_k$ minus 1.
A $\alg$-LRS is \emph{simple} if all its characteristic roots have
multiplicity 1, and is \emph{non-degenerate} if no quotient of two
distinct characteristic roots is a root of unity.

A characteristic root $\lambda$ is \emph{dominant} if $|\lambda| \ge |\mu|$ for all characteristic roots $\mu$.
We define the \emph{dominant part} $\langle v_n\rangle_{n=0}^\infty$ of $\langle
u_n\rangle_{n=0}^\infty$ by writing
\begin{equation*}
    v_n = \sum_{\text{$\lambda_j$ dominant}} Q_j(n)\lambda_j^n \, ,
\end{equation*}
with the \emph{non-dominant part} $\langle r_n \rangle_{n=0}^\infty$
given by $\langle u_n - v_n\rangle_{n=0}^\infty$.
One can compute positive real numbers $r$ and $R$ such that $rR^n >
|r_n|$ for all $n \in \nat$ and $R < |\lambda|$ for $\lambda$ a
dominant characteristic root, see, e.g.,~\cite[Lem.~2.5]{berthe2024decidability}.

For every $M \ge 2$ and $\mathbb{Z}$-LRS $\langle
u_n\rangle_{n=0}^\infty$,
the sequence $\langle u_n\bmod M\rangle_{n=0}^\infty$
is ultimately periodic and one can effectively
compute its period and
preperiod~\cite[Lem.~2.6]{berthe2024decidability}.

\begin{example}\label{example 1}
    Let $\langle u_n\rangle_{n=0}^\infty$ be the ($\mathbb{Z}$-)LRS
    \eqref{eq : example LRS} from Sec.~\ref{sec : intro}.  
    The characteristic polynomial of $\langle u_n \rangle_{n=0}^\infty$ is 
    \begin{equation*}
        F(X) = X^3 - 6X^2 + 13X - 10 = (X-2)\,(X - (2+i))\,(X - (2-i)) \, . 
    \end{equation*}
    The characteristic roots of $\langle u_n\rangle_{n=0}^\infty$ are
    $2$, $2+i$, and $2-i$.
    Using linear algebra, one easily recovers the
    exponential-polynomial representation~\eqref{eq : example LRS poly exp form}. 
    The LRS $\langle u_n\rangle_{n=0}^\infty$ is simple and non-degenerate. 
    Its dominant part is the sequence defined by $v_n =
    \frac{1}{2}(2+i)^n + \frac{1}{2}(2-i)^n$,
    while its non-dominant part is given by $r_n = 2^n$. 
\end{example}
\subsection{Number Theory}\label{sec : number theory}

Our main number-theoretic tool is Baker's theorem on linear forms in logarithms.  Specifically, we make use of a flexible
version due to Matveev along with some off-the-shelf applications due to
Mignotte, Shorey, and Tijdeman.

Let $\alpha \neq 0$ be an algebraic number of degree $d$ with minimal polynomial $F(X) = a_0 \prod_{i=1}^d(X-\alpha_i)$. 
The \emph{logarithmic Weil height} of $\alpha$ is defined as
\begin{equation*}
    h(\alpha) = \frac{1}{d}\left(\log|a_0| + \sum_{i=1}^d \log
      \max\{|\alpha_i|, 1\} \right) \, .
\end{equation*}
Furthermore, set $h(0) = 0$.
For all algebraic numbers $\alpha_1,\dots,\alpha_k$ and $n \in \mathbb{Z}$, we have the following properties:
\begin{align*}
    h(\alpha_1+\cdots+\alpha_k) &\le \log k + h(\alpha_1) + \cdots +
                                  h(\alpha_k) \, ,\\
  h(\alpha_1\alpha_2) &\le h(\alpha_1) + h(\alpha_2) \, , \quad \text{and}\\
    h(\alpha_1^n) &\le |n|h(\alpha_1) \, .
\end{align*}
A \emph{number field} $L$ is a field extension of $\mathbb{Q}$ such that the degree of $L/\mathbb{Q}$ is finite. 
If $\alpha_1,\dots,\alpha_d \in \alg$, $L = \mathbb{Q}(\alpha_1,\dots,\alpha_d)$ is a number field whose degree can be effectively computed.

Let $L$ be a number field of degree $D$, $M \ge 1$, $\alpha_1,\dots,\alpha_M \in L^*$, and $b_1,\dots,b_M \in \mathbb{Z}$. Then set $B =\max\{|b_1|,\dots,|b_M|\}$,
\begin{equation*}
    \Lambda = \prod_{j = 1}^M \alpha_j^{b_j} - 1 \, , \quad \text{and} \quad
    h'(\alpha_j) = \max\big\{Dh(\alpha_j), |\log{\alpha_j}|, 0.16\big\} \, .
\end{equation*}
Matveev~\cite{matveev2000explicit} proved the following:
\begin{theorem}\label{thm : Matveev}
    If $\Lambda \ne 0$, then
    \begin{equation*}
        \log{|\Lambda|} > -3  \cdot 30^{M+4}(M+1)^{5.5}D^2(1+\log{D})\cdot \big(1 + \log(MB)\big)\prod_{j=1}^d h'(\alpha_j) \, .
    \end{equation*}
    In particular, there is a computable constant $c$ such that when $\Lambda \ne 0$, $|\Lambda| > B^{-c}$.
\end{theorem}

We will also need the following results from Mignotte, Shorey, and Tijdeman~\cite{shorey1984distance}:
\begin{theorem}\label{thm : MST size LRS}
    Let $\langle u_n \rangle_{n=0}^\infty$ be a non-degenerate LRS
    with two  dominant roots of magnitude $|\lambda|$. Then there are
    computable positive constants $C_1$ and $C_2$ such that
    \begin{equation*}
        |u_n| \ge |\lambda|^n \cdot n^{-C_1\log(n)}
    \end{equation*}
    whenever $n \ge C_2$.
\end{theorem}
\begin{theorem}\label{thm : MST size difference terms of LRS}
    Let $\langle u_n \rangle_{n=0}^\infty$ be a non-degenerate LRS
    with two  dominant roots of magnitude $|\lambda|$. Then there are
    computable positive constants $C_3$ and $C_4$ such that
    \begin{equation*}
        |u_{n_1} - u_{n_2}| \ge |\lambda|^{n_1} \cdot {n_1}^{-C_3\log({n_1})\log({n_2}+2)}
    \end{equation*}
    whenever ${n_1} > {n_2}$ and ${n_1} \ge C_4$.
\end{theorem}

\section{Proof of the Main Result}

In this section, we prove Thm.~\ref{thm : MSO decidable}: the
MSO theory of the structure $\langle \nat ; < , P\rangle$ is decidable
whenever $P$ is a predicate comprising the set of positive terms of
some non-degenerate, simple, integer-valued LRS having two dominant
characteristic roots.

We begin in Sec.~\ref{sec : Reduction to proving weak normality} by
untangling the definition of an LRS satisfying the above hypotheses
and reduce Thm.~\ref{thm : MSO decidable} to Thm.~\ref{thm : intro
  weak normality}.  We then provide intuition underlying the proof of
the latter through an extended example in Sec.~\ref{sec : example}, and
consider a continuous version of our problem in Sec.~\ref{sec : Proof of
  weak normality}.  Finally, in the same section, we establish
Thm.~\ref{thm : intro weak normality}.

\subsection{Reduction to Prodisjunctivity}\label{sec : Reduction to proving weak normality}

Let $\langle u_n\rangle_{n=0}^\infty$ be an LRS satisfying the
hypotheses of Thm.~\ref{thm : MSO decidable}. We first record some elementary observations whose proof is in App.~\ref{sec : omitted proofs}.
\begin{lemma}\label{lem : dominant roots are nice}
  Assume that $\langle u_n\rangle_{n=0}^\infty$ is an LRS satisfying the hypotheses of
  Thm.~\ref{thm : MSO decidable} and whose two dominant roots are $\lambda_1$ and
  $\lambda_2$. Then $\lambda_2 = \overline{\lambda_1}$,
the argument of $\lambda_1$ is not a rational multiple of $\pi$, and $|\lambda_1| > 1$.
\end{lemma}

When writing $\langle u_n\rangle_{n=0}^\infty$ in its
exponential-polynomial form, we split it into its dominant part
$\langle v_n\rangle_{n=0}^\infty$ and non-dominant part $\langle r_n\rangle_{n=0}^\infty$:
\begin{equation}\label{eq : def u_n and v_n}
    u_n = v_n + r_n = \alpha\lambda^n +
    \overline{\alpha}\overline{\lambda}^n + r_n \, .
\end{equation}
Here, $\alpha$ and $\lambda$ are algebraic numbers with $\alpha \ne 0$
(as otherwise $\lambda$ and $\overline{\lambda}$ would not be
characteristic roots, i.e., roots of the polynomial corresponding to
the \emph{minimal} recurrence relation that
$\langle u_n\rangle_{n=0}^\infty$ obeys), $|\lambda| > 1$, and
the argument of $\lambda$ is not a rational multiple of $\pi$.
We keep $\lambda$, $\alpha$, and $\langle r_n\rangle_{n=0}^\infty$ fixed for the remainder of the paper.

Recall that $P = \{u_n \colon n \in \nat\} \cap \nat$ and
$\langle p_m\rangle_{m=0}^\infty$ is an enumeration of $P$ in
increasing order.
To apply Prop.~\ref{prop : reduction one unary predicate}, we need the following lemma, whose proof relies heavily on the results of Mignotte, Shorey, and Tijdeman from Sec.~\ref{sec : LRS} and can be found in App.~\ref{sec : omitted proofs}.
\begin{lemma}\label{lem : First three conditions of LRS}
Let $P \subseteq \nat$ be as per Thm.~\ref{thm : MSO decidable}.
Then $P$ is infinite, recursive, and effectively sparse.
\end{lemma}

By Prop.~\ref{prop : reduction one unary predicate}, it now
suffices to prove that for all $M \ge 2$, one can decide
$\Acc_{\langle p_m\bmod M\rangle_{m=0}^\infty}$ (determine
whether a given deterministic Muller automaton $\mathcal{A}$ over
alphabet $\{0,\dots,M-1\}$ accepts
$\langle p_m\bmod M\rangle_{m=0}^\infty$).  The next lemma
shows how the effective prodisjunctivity of
$\langle p_m \rangle_{m=0}^\infty$ asserted by~Thm.~\ref{thm : intro weak normality} enables us to do this.
Its proof can be found in App.~\ref{sec : omitted proofs}.

\begin{lemma}\label{lem : Theorems intro implications}
    Theorem~\ref{thm : intro weak normality} implies Thm.~\ref{thm : MSO decidable}, i.e., if $\langle u_n\rangle_{n=0}^\infty$ is an LRS satisfying the conditions of Thm.~\ref{thm : MSO decidable}, $P := \{u_n \colon n\in\nat\}\cap\nat$, and the enumeration of $P$ is prodisjunctive, then the MSO theory of the structure $\langle \nat ; < , P\rangle$ is decidable.
\end{lemma}

Theorem~\ref{thm : intro weak normality} asserts that, for any
$M \geq 2$, the sequence
$\langle p_m\bmod M\rangle_{m=N_M}^\infty \in {S_M}^\omega$
is disjunctive.  Let us unpack this definition.
We have that $\langle p_m\bmod M\rangle_{m=N_M}^\infty \in {S_M}^\omega$ is
disjunctive if for any $\ell \ge 1$, every pattern
$\langle s_1,\dots,s_\ell\rangle \in {S_M}^\ell$ appears infinitely often in
$\langle p_m\bmod M\rangle_{m=N_M}^\infty$.  That is, for
all $\ell, N \in \nat$ with $N \geq N_M$ and
$s_1,\dots,s_\ell \in S_M$, there are $n_1,\dots,n_\ell \in \nat$ such
that
\begin{enumerate}
    \item $n_1,\dots,n_\ell \ge N$;
    \item for all $1 \le i \le \ell$, $u_{n_i} \equiv s_i \pmod M$;
    \item $0 \le u_{n_1} < \cdots < u_{n_\ell}$;
    \item for all $m \ge 0$ such that $u_{n_1} \le u_m \le u_{n_\ell}$, $u_m \in \{u_{n_1},\dots,u_{n_\ell}\}$. 
\end{enumerate}

As the dominant part $\langle v_n \rangle_{n=0}^\infty$ (where $v_n = \alpha\lambda^n+\overline{\alpha}\overline{\lambda}^n$, see~\eqref{eq : def u_n and v_n}) only relies on two algebraic numbers, $\alpha$ and $\lambda$, it is easier to work
with $\langle v_n\rangle_{n=0}^\infty$ than $\langle u_n\rangle_{n=0}^\infty$.
In the following lemma, we prove that we can actually effect this change.
\begin{lemma}\label{lem : weak normality equivalent statement}
    Assume that for all natural numbers $\ell$, $N$, $T \ge 2$, and $t_1,\dots,t_\ell \in \{0, \dots,T-1\}$, there are $n_1,\dots, n_\ell \in \nat$ such that
    \begin{enumerate}
        \item $n_1,\dots, n_\ell \ge N$;
        \item for all $1 \le j \le \ell$, $n_j \equiv t_j \pmod T$;
        \item $0 < v_{n_1} < \cdots < v_{n_\ell}$;
        \item for all $m \geq 0$ such that $v_{n_1} < v_m < v_{n_\ell}$, $m \in \{n_2,\dots,n_{\ell-1}\}$.
    \end{enumerate}
    Then the enumeration $\langle p_m \rangle_{m=0}^\infty$ of $P := \{u_n \colon n \in \nat\} \cap \nat$ is effectively prodisjunctive.
\end{lemma}
\begin{proof}
    We claim that for some computable number $N'$, $u_{m_1} < u_{m_2}$ if and only if $v_{m_1} < v_{m_2}$ whenever $m_1, m_2 \ge N'$. 
Suppose not. Then, without loss of generality, $m_1 > m_2$ and $|v_{m_1} - v_{m_2}| < |r_{m_1}| + |r_{m_2}| < 2rR^{m_1}$. 
   But Thm.~\ref{thm : MST size difference terms of LRS} implies that for $m_1 \ge C_4$,
    \begin{align*}
        \log(2r) + m_1\log R &> \log|v_{m_1} - v_{m_2}|\\
        &> m_1 \log|\lambda| - C_3 \log(m_1)^2\log(m_2+2)\\
        &> m_1 \log|\lambda| - C_3 \log(m_1+1)^3 \, ,
    \end{align*}
    which cannot hold for $m_1 \ge N'$ for some computable $N' \in
    \nat$ as $\log|\lambda| > \log|R|$. Our claim therefore follows.

    Assume that $\langle u_n\bmod M\rangle_{n=0}^\infty$ has
    period $T$ (which can be effectively computed).  If
    $s_1,\dots,s_\ell \in S$ and $1 \le j \le \ell$, there exists a
    $t_j \in \{0,\dots,T-1\}$ such that
    $u_{nT +t_j} \equiv s_j \pmod M$ whenever $n$ is large enough.
    Therefore, if $n_1,\dots,n_\ell \in \nat$ are at least $\max\{N,N'\}$
    and satisfy the hypotheses of the lemma, it follows that
    $0 \le u_{n_1} < \cdots < u_{n_\ell}$ and for all $m \in \nat$
    such that $u_{n_1} < u_m < u_{n_\ell}$,
    $m \in \{n_2,\dots,n_{\ell-1}\}$.  In other words, if $p_m = u_{n_1}$,
    then for $2 \le j \le \ell$, $p_{m+j-1} = u_{n_j}$ and
    $p_{m+j-1} = u_{n'T +t_j} \equiv s_j \pmod M$ for some
    $n' \in \nat$.
\end{proof}

As $|\alpha| \ne 0$ and Lem.~\ref{lem : weak normality equivalent
  statement} is only concerned with inequalities $v_{n_1} < v_{n_2}$ and $v_{n_1} > 0$ for natural numbers ${n_1}$ and ${n_2}$, we can scale $v_n$ by $1/|2\alpha|$. 
That is, we can assume that $|\alpha| = 1/2$.
Write $\alpha = \frac{1}{2}e^{i\phi}$ and $\lambda=|\lambda|e^{i\theta}$.
Thus, $v_n = \cos(\theta n + \phi)|\lambda|^n$.

Let us now sketch the method we will use to establish the hypothesis of Lem.~\ref{lem : weak normality equivalent statement} whose proof relies heavily on the following notation.
\begin{definition}\label{definition J}
    For integers $d \ne 0$ and real numbers $0 < \gamma < \delta$, define $\mathcal{J}_d(\gamma, \delta) \subset \mathbb{R}/(2\pi\mathbb{Z})$ as 
    \begin{equation*}
        \mathcal{J}_d(\gamma, \delta) = \Big\{x \in \mathbb{R}/(2\pi\mathbb{Z}): 
        0 < \gamma\cos(x) < \cos(x + d\theta)|\lambda|^d < \delta\cos(x)\Big\} \, .
      \end{equation*}
Moreover, we define  $\mathcal{J}_d(0,\delta)$ to stand for the limit of $\mathcal{J}_d(\gamma,\delta)$ as $\gamma$ tends to $0$.
  \end{definition}
  
In Lem.~\ref{lem : upper bound length intervals}, we can quantify the size of these intervals. 
For a real number $x$, let $|x|_{2\pi}$ denote the distance of $x$ to the nearest integer multiple of $2\pi$.
The proof of Lem.~\ref{lem : upper bound length intervals} can be found in App.~\ref{sec : omitted proofs}.
\begin{lemma}\label{lem : upper bound length intervals}
    One can compute constants $C_5$, $C_6$, $C_7$, and $C_8$ such that for all $0 < \gamma < \delta < \sqrt{|\lambda|}$ and all $d \ge 1$, $\mathcal{J}_d(\gamma,\delta)$ consists of a single interval,
    \begin{equation*}
        C_6\frac{\delta - \gamma}{|\lambda|^dd^{C_7}} <
        |\mathcal{J}_d(\gamma,\delta)| < C_5\frac{\delta -
          \gamma}{|\lambda|^d} \, ,
    \end{equation*}
    and $|x-(-d\theta \pm \pi/2)|_{2\pi} < C_8 |\lambda|^{-d}$ for any
    $x \in \mathcal{J}_d(\gamma,\delta)$.
\end{lemma}

We now continue with the most technical step in the proof of Thm.~\ref{thm : MSO decidable}: establishing a continuous version of the hypothesis of Lem.~\ref{lem : weak normality equivalent statement}.
Items 3 and 4 of Lem.~\ref{lem : weak normality equivalent statement} involve terms of the form $v_n$, and we can divide them by $|\lambda|^{n_1}$. 
Further, for $2 \le j \le \ell$, we set $b_j := n_j - n_1$ and $x = n\theta + \phi$.
Then, item 3 turns into
\begin{equation}\label{eq : item 3 continuous}
    0 < \cos(x) < \cos(x + b_2\theta)|\lambda|^{b_2} < \cdots < \cos(x + b_\ell\theta)|\lambda|^{b_\ell}
\end{equation}
and item 4 turns into 
\begin{equation}\label{eq : item 4 continuous}
    \forall d \in \mathbb{Z} \colon \cos(x) < \cos(x + d\theta)|\lambda|^d < \cos(b_\ell\theta)|\lambda|^{b_\ell} \Longrightarrow d \in \{ b_2,\dots,b_{\ell-1}\}.
\end{equation}
We want to find an interval $\mathcal{I} \subset \mathbb{R}/(2\pi\mathbb{Z})$ and numbers $b_j$ such that $b_j \equiv t_j - t_1 \pmod{T}$ (for item 2),~\eqref{eq : item 3 continuous} holds for all $x \in \mathcal{I}$, and~\eqref{eq : item 4 continuous} hold for ``most'' $x \in \mathcal{I}$.
This leads us to the following:
\begin{lemma}\label{lem : density weakly normal}
    Let $\ell, T \ge 2$ and $t_1, \dots, t_\ell \in \{0,\dots,T-1\}$. Then there are an interval $\mathcal{I} \subset \mathbb{R}/(2\pi\mathbb{Z})$, natural numbers $D, b_2,\dots,b_\ell \ge 1$, and real numbers $1 < \delta_2 < \cdots < \delta_\ell < \sqrt{|\lambda|}$ such that
    \begin{enumerate}
        \item[(a)] for all $2 \le j\le \ell$, $b_j \equiv t_j - t_1 \pmod T$;
        \item[(b)] for all $x \in \mathcal{I}$,
        \begin{align}\label{eq : cos seperated result}
            \begin{split}
                0 < \cos(x) &< \cos(x + b_2\theta)|\lambda|^{b_2} < \delta_2 \cos(x) \\
                & < \cos(x + b_3\theta)|\lambda|^{b_3} < \delta_3 \cos(x) \\
                &\quad\quad \vdots \\
                & < \cos(x + b_\ell\theta )|\lambda|^{b_\ell} < \delta_\ell \cos(x)\,;
            \end{split}
        \end{align}
        \item[(c)] for all integers $d < D$ not in the set $\{b_2,\dots,b_{\ell-1}\}$, $\mathcal{I} \cap \mathcal{J}_d(1, \delta_\ell) = \emptyset$;
        \item[(d)] $ \sum_{d = D}^\infty |\mathcal{J}_d(1, \delta_\ell)| < |\mathcal{I}|$.
    \end{enumerate}
\end{lemma}
Translating (b) in the notation of intervals $\mathcal{J}$, we have that
$$
\mathcal{I} \subseteq (-\pi/2,\pi/2) \cap \mathcal{J}_{b_2}(1, \delta_2) \cap \mathcal{J}_{b_3}(\delta_2, \delta_3)\cap \cdots \cap \mathcal{J}_{b_3}(\delta_{\ell-1}, \delta_\ell).
$$

The proof of Lem.~\ref{lem : density weakly normal} is in App.~\ref{sec : omitted proofs}.
The remaining tools we need for this proof are listed in Sec.~\ref{sec
  : Proof of weak normality}.
Finally, we derive Thm.~\ref{thm : MSO decidable} from Lem.~\ref{lem : density weakly normal}.
But first, in the upcoming section we go through an example to recap the construction.

\subsection{An Extended Example}\label{sec : example}

Let $\langle u_n\rangle_{n=0}^\infty$ be the sequence \eqref{eq :
  example LRS} from Sec.~\ref{sec : intro} and assume $M = 5$.
Let us study the sequence $u_n = \frac{1}{2}(2+i)^n + \frac{1}{2}(2-i)^n + 2^n$ modulo $5$.

From~\eqref{eq : example intro modular behavior}, we conclude that
$S_5 = \{0,2,4\}$ and $T = 4$ (that is, $\langle u_n \bmod 5\rangle_{n=0}^\infty$ is ultimately periodic with period $T=4$, and $\{0, 2, 4\}$ are exactly the congruence classes that appear infinitely often in $\langle u_n \bmod 5\rangle_{n=0}^\infty$).
Thus for all $m \ge 0$, $p_m\bmod 5 \in \{0,2,4\}$, where $\langle p_m\rangle_{m=0}^\infty$ is the enumeration of $P := \{u_n \colon n \in \nat\} \cap \nat$.
Then Thm.~\ref{thm : intro weak normality} asserts that every
$\langle s_1,\dots,s_\ell\rangle \in {S_5}^*$ appears in
$\langle p_m \bmod 5\rangle_{m=0}^\infty$ infinitely often as a factor. 
We will show that this indeed holds when $\ell = 3$ and $\langle s_1,s_2,s_3 \rangle = \langle 0,0,0 \rangle$.

Let us compute $t_1$, $t_2$, and $t_3$ (the conjugacy classes modulo $T = 4$ such that for large enough $n$, $u_{Tn + t_i} \equiv s_i \bmod 5$). 
In view of \eqref{eq : example intro modular behavior}, we are forced to take $t_i = 3$ for $i = 1,2,3$ as $s_i = 0$.
Thus, to find $\langle0,0,0\rangle$ in $\langle
p_m\rangle_{m=0}^\infty$, we want to find $n_1,n_2,n_3 \in \nat$
congruent to $3$ modulo $4$ such that $0 < u_{n_1} < u_{n_2} <
u_{n_3}$, and if $u_{n_1} < u_m < u_{n_3}$ for some $m \ge 0$, then
$m \in \{n_2\}$ (i.e., $m = n_2$).

By \eqref{eq : example LRS poly exp form}, the dominant part $\langle v_n\rangle_{n=0}^\infty$ of $\langle u_n\rangle_{n=0}^\infty$ is given by $v_n = \frac{1}{2}(2+i)^n + \frac{1}{2}(2-i)^n$ and the non-dominant part $\langle r_n\rangle_{n=0}^\infty$ by $r_n = 2^n$.
As shown in Lem.~\ref{lem : weak normality equivalent statement}, we can work with $v_n = \alpha \lambda^n + \overline{\alpha} \overline{\lambda}^n = \cos(n\theta + \phi)|\lambda|^n$ instead of $u_n$ for large enough~$n$.  
Here, we have that $\lambda = e^{i\theta}|\lambda| = 2+i$ and $\alpha = \frac{1}{2}e^{i\phi} = \frac{1}{2}$.
Thus, in this example, $\phi = 0$.

Hence, for a given $N \in \nat$ (for simplicity, let us take $N = 0$) we wish to find $n_1,n_2,n_3 \ge N$ that are congruent to $3$ modulo $4$, satisfy
\begin{equation*}
  0 < \cos(n_1\theta) < \cos(n_2\theta )|\lambda|^{n_2-n_1} <
  \cos(n_3\theta)|\lambda|^{n_3-n_1} \, ,
\end{equation*}
and such that $m = n_2$ for any $m \in \nat$ satisfying the inequality $\cos(n_1\theta) < \cos(m\theta)|\lambda|^{m-n_1} < \cos(n_3 \theta)|\lambda|^{n_3-n_1}$.
This is the statement of Lem.~\ref{lem : weak normality equivalent statement}.  

Next, we aim to find $b_2, b_3 \in \nat$ such that $n_1 = n$, $b_2 = n_2 - n$, and $b_3 = n_3 - n$ for some $n \ge N$ congruent to $3$ modulo $4$ that satisfies our hypotheses.
To ensure that for $j \in \{2,3\}$, $n_j \equiv t_j \equiv 3 \pmod{4}$, we have to require that $b_j \equiv 0 \pmod{4}$. 

In order to solve this discrete problem (find natural numbers $n_1,n_2,n_3$ meeting these constraints), we first want to solve a continuous variant of this problem: find $b_2$ and $b_3$ and an \emph{interval} $\mathcal{I} \subset \mathbb{R}/(2\pi\mathbb{Z})$ such that for $x \in\mathcal{I}$, these properties hold ``often''.  
We will construct an open, non-empty interval $\mathcal{I} \subset \mathbb{R}/(2\pi\mathbb{Z})$ such that for all $x \in \mathcal{I}$,
\begin{equation}\label{eq : inequalilty x example}
    0 < \cos(x) < \cos(x + b_2\theta)|\lambda|^{b_2} < \cos(x + b_3\theta)|\lambda|^{b_3}\,.
\end{equation}
We cannot ensure that for all $x \in \mathcal{I}$ and $m \in \mathbb{Z}$, $\cos(x) < \cos(x + m\theta)|\lambda|^m < \cos(x + b_3\theta)|\lambda|^{b_3}$ implies that $m = b_2$. 
However, we can ensure that it happens for ``many'' $x \in \mathcal{I}$, including for infinitely many numbers of the form $x = n\theta$, where $n \equiv 3 \pmod 4$.

We can translate the last part into the notation of the intervals  $\mathcal{J}_d(\gamma, \delta)$.
Then, arbitrarily setting $\delta_2 = 1.95$ and $\delta_3 = 2$, we strengthen~\eqref{eq : inequalilty x example} in the form of item (b) in Lem.~\ref{lem : density weakly normal}: we want that for all $x \in \mathcal{I}$,
\begin{equation*}
    0 < \cos(x) < \cos(x + b_2\theta)|\lambda|^{b_2} <1.95 \cos(x) < \cos(x + b_3\theta)|\lambda|^{b_3} < 2 \cos(x)\,.
\end{equation*}
Thus, $I \subseteq (-\pi/2,\pi/2)$ (as $\cos(x) > 0$ and $x \in \mathcal{I}$), $I \subseteq \mathcal{J}_{b_2}(1, 1.95)$ (as $1\cdot \cos(x) < \cos(x + b_2\theta)|\lambda|^{b_2} <1.95 \cos(x)$), and similarly, $I \subseteq \mathcal{J}_{b_3}(1.95, 2)$.

Dealing with items (c) and (d) of Lem.~\ref{lem : density weakly normal} is harder. 
In Lem.~\ref{lem : upper bound length intervals}, we have proven many useful results on the structure of the sets $\mathcal{J}_d(\gamma, \delta)$.
For $d \in \mathbb{Z}$ and small enough $\delta$, if $\mathcal{J}_d(\gamma, \delta)$ is non-empty, it is a single interval and $|\mathcal{J}_d(\gamma, \delta))| = O((\delta-\gamma)|\lambda|^{-d})$.
Hence, item (d) of Lem.~\ref{lem : density weakly normal} can easily be estimated using a geometric series:
\begin{equation*}
    \sum_{d = D}^\infty |\mathcal{J}_d(1, \delta_\ell)| \le O(\delta_\ell|\lambda|^{-D})\,.
\end{equation*}
Moreover, every point in $\mathcal{J}_d(\gamma, \delta))$ is at most $O(\delta|\lambda|^{-d})$ away from $-d\theta \pm \pi/2$ in $\mathbb{R}/(2\pi\mathbb{Z})$.
We illustrate this in Fig.~\ref{fig : examples of intervals}.

\begin{figure}
    \begin{tikzpicture}[scale=2.3]
    \draw (-1.57079,0) -- (1.57079,0);
    \draw (-1.57079, -0.2) -- (-1.57079, 0.25);
    \draw (1.57079, -0.2) -- (1.57079, 0.25);
    \node at (-1.57079, -0.3) {$-\pi/2$};
    \node at (1.57079, -0.3) {$\pi/2$};
    \foreach \i in {-1,...,1}
    {
        \draw[ultra thin] (\i, 0.25) -- (\i, -0.075);
        \node at (\i, -0.15) {\small{$\i$}};
    }
\draw[fill =black, color = black] (1.10714871779409,0.000000000000000) circle (0.03);
\draw [cyan,ultra thick] plot (0.785398163397448,0.000000000000000) -- (-0.785398163397448,0.000000000000000);
\draw[fill =black, color = black] (0.643501108793284,0.0750000000000000) circle (0.03);
\draw [Green,ultra thick] plot (0.463647609000806,0.0750000000000000) -- (0.000000000000000,0.0750000000000000);
\draw[fill =black, color = black] (0.179853499792478,0.150000000000000) circle (0.03);
\draw [Magenta,ultra thick] plot (0.0906598872007451,0.150000000000000) -- (-0.0906598872007451,0.150000000000000);
\draw[fill =black, color = black] (-0.283794109208328,0.225000000000000) circle (0.03);
\draw [red,ultra thick] plot (-0.321750554396642,0.225000000000000) -- (-0.394791119699762,0.225000000000000);
\end{tikzpicture}
    \caption{Let $\lambda = 1+2i$. Then for $d = 1,2,3,4$, $\mathcal{J}_d(1, 3)$ are drawn in cyan, green, magenta, and red, respectively, and $|\mathcal{J}_d(1, 3)|$ are $\pi/2$, $0.464$, $0.182$, $0.073$, respectively. 
    As some of the intervals overlap, we have stacked them vertically for visual purposes. The points in black mark out $-d\theta \pm\pi/2$ for the corresponding value of $d$.}
    \label{fig : examples of intervals}
\end{figure}
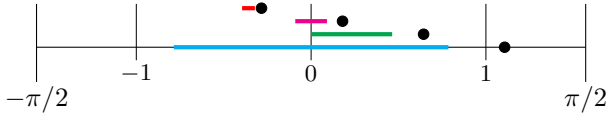

As we have guessed $\delta_2 = 1.95$ and $\delta_3 = 2$, we will construct $\mathcal{I}$ inductively as 
\begin{equation*}
    \mathcal{I} := \mathcal{I}_3 := \mathcal{J}_{b_3}(1.95,2) \subseteq \mathcal{I}_2 := \mathcal{J}_{b_2}(1,1.95)\subseteq \mathcal{I}_1 \, .
\end{equation*}
At each step $i=1,2,3$, the items (a) and (b) hold for $j \le i$ while items (c) and (d) also hold for the interval $\mathcal{I}_i$.

First, for $i = 1$, we take $\mathcal{I}_1 := (-1.1, 1.1) \subset (-\pi/2, \pi/2)$.
Then, for an integer $d < 0$ and $x \in \mathcal{I}_1$, we have that $\cos(x + \theta d)|\lambda|^{d} < \cos(x)$.
So in item (c), we can take $D = 1$. 
As $\delta_3 = 2$,
we can prove that item (d) is satisfied:
\begin{equation}\label{eq : strong inequality sizes J_d}
    \Big|\bigcup_{d = 1}^\infty \mathcal{J}_d(1, \delta_3)\Big| \le \sum_{d = 1}^\infty |\mathcal{J}_d(1, \delta_3)| <  |\mathcal{I}_1| \, .
\end{equation}

For $i = 2$, recall that $\delta_2 = 1.95$ and that $b_2$ has to satisfy $\mathcal{J}_{b_2}(1, 1.95) \subset \mathcal{I}_1$ and $b_2 \equiv 0 \pmod 4$. 
We pick $b_2 = 4$, being the smallest possible choice for $b_2$.
As shown in Fig.~\ref{fig : example interval 2}, we have that $\mathcal{I}_2$ is indeed in $\mathcal{I}_1$.
For all integers $d \le 20$ not equal to either $0$ or $4$, $\mathcal{I}_2 \cap \mathcal{J}_d(1, 2)$ is empty.
As $\sum_{d = 21}^\infty |\mathcal{J}_d(1, 2)| < |\mathcal{J}_4(1,
1.95)|$, $\mathcal{I}_2 := \mathcal{J}_4(1, 1.95)$ is not covered by
the intervals $\mathcal{J}_d(1, 2)$ with $d \notin \{0, 4\}$.
Hence, $D = 21$ is valid choice in this step of our inductive procedure.

\begin{figure}
    \begin{tikzpicture}[scale=2.52]
    \draw (-1.5708,-0.15) -- (1.5708,-0.15);
    \draw (-1.5708, -0.175) -- (-1.5708, 0);
    \draw (1.5708, -0.175) -- (1.5708, 0);
    \node at (-1.5708, -0.225) {$-\pi/2$};
    \node at (1.5708, -0.225) {$\pi/2$};
    \draw [thin, dotted] plot (-0.3218, 0) -- (-0.3218, -0.225);
    \draw [thin, dotted] plot (-0.3588, 0) -- (-0.3588, -0.225);
    \node at (-0.3403, -0.325) {$\mathcal{J}_4(1,1.95)$};
    \draw [thin, dotted] plot (-1.1, 0.15) -- (-1.1, -0.2);
    \draw [thin, dotted] plot (1.1, 0.15) -- (1.1, -0.2);
    \draw [thin, dotted] plot (-1.1, 0.15) -- (1.1, 0.15);
    \node at (0, 0.25) {$\mathcal{I}_1$};
\draw [Red, ultra thick] plot (0.7854, -0.03) -- (0.0, -0.03);
\draw [Red, ultra thick] plot (0.4636, -0.06) -- (0.245, -0.06);
\draw [Red, ultra thick] plot (0.0907, -0.09) -- (0.0, -0.09);
\draw [Green, ultra thick] plot (-0.3218, -0.12) -- (-0.3588, -0.12);
\draw [fill =Red, color =Red] (-0.7794,-0.12) circle (0.01);
\draw [fill =Red, color =Red] (-1.218,-0.12) circle (0.01);
\draw [fill =Red, color =Red] (1.4674,-0.12) circle (0.01);
\draw [fill =Green, color =Green] (1.0045,-0.12) circle (0.01);
\draw [fill =Red, color =Red] (0.5405,-0.12) circle (0.01);
\draw [fill =Red, color =Red] (0.0764,-0.12) circle (0.01);
\draw [fill =Red, color =Red] (-0.3875,-0.12) circle (0.01);
\draw [fill =Green, color =Green] (-0.8513,-0.12) circle (0.01);
\draw [fill =Red, color =Red] (-1.315,-0.12) circle (0.01);
\draw [fill =Red, color =Red] (1.3629,-0.12) circle (0.01);
\draw [fill =Red, color =Red] (0.8993,-0.12) circle (0.01);
\draw [fill =Green, color =Green] (0.4356,-0.12) circle (0.01);
\draw [fill =Red, color =Red] (-0.028,-0.12) circle (0.01);
\draw [fill =Red, color =Red] (-0.4917,-0.12) circle (0.01);
\draw [fill =Red, color =Red] (-0.9553,-0.12) circle (0.01);
\draw [fill =Green, color =Green] (-1.419,-0.12) circle (0.01);
\end{tikzpicture}
    \caption{We drew the intersection of $\mathcal{I}_1$ and $\mathcal{J}_d(1, 1.95)$ for $d = 1,\dots,20$ in red when $d \not\equiv 0 \pmod 4$ and in green when $d\equiv 0 \pmod 4$. 
    For ease of visibility, $\mathcal{J}_d(1, 2)$ is positioned
    higher for $d=1,2,3$, and for intervals $\mathcal{J}_d(1, 2)$ that
    are too small to draw, their position is marked out with a dot.}
    \label{fig : example interval 2}
\end{figure}
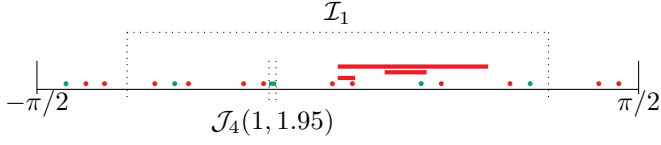

For $b_3$, recall that $b_3 \equiv 0 \pmod{4}$.
We find that $\mathcal{J}_b(1.95, 2) \cap \mathcal{I}_2$ is non-empty for $b = 0,4,38, 99, 160, 309, 370,\dots$. 
The smallest such $b$ that is not equal to $0$ or $4$ (which are
already in use) and congruent to $0$ modulo $4$ is $160$.
Then,
\begin{equation*}
    \sum_{d=1, d\notin \{4, 160\}}^\infty \big|\mathcal{J}_d(1, 2) \cap
    \mathcal{J}_{160}(1.95, 2)\big| < \big|\mathcal{J}_{160}(1.95, 2)\big| 
\end{equation*}
allows us to take $D  =160$ and $\mathcal{I} := \mathcal{J}_{160}(1.95, 2)$.
This interval $\mathcal{I}$ is tiny: it has length approximately $5.7 \cdot 10^{-58}$.
However, it satisfies the constraints in Lem.~\ref{lem : density weakly normal}.

Now we must show that there is some $x = n\theta \in \mathcal{I}$ such
that $n \equiv 3 \pmod 4$ and item~4 of Lem.~\ref{lem : weak
  normality equivalent statement} holds, as the three other conditions
are already satisfied.

Recall that for a real number $x$, $|x|_{2\pi}$ denotes the distance from $x$ to the nearest integer multiple of $2\pi$.
If $x = n\theta$ and $n\ge N$, then items 1, 2 and 3 of Lem.~\ref{lem : weak normality equivalent statement} are satisfied. 
Assume that for $n$ as above item 4 in Lem.~\ref{lem : weak normality equivalent statement} is not satisfied. 
Then $n\theta \in \mathcal{J}_d(1, 2)$ for some $d \in \mathbb{Z}\setminus\{0, 4, 160\}$. 
Hence, combining the information from Lem.~\ref{lem : upper bound length intervals} to the effect that $\mathcal{J}_d(1, 2)$ is close to $-d\theta \pm \pi/2$ modulo $2\pi$, we have:
\begin{equation}\label{eq : example Baker end}
    c_1 |\lambda|^{-d} > |n\theta  - (-d\theta \pm \pi/2)|_{2\pi} > |d + n|^{-c_2}
\end{equation}
for two constants $c_1$ (derived from Lem.~\ref{lem : upper bound length intervals}) and $c_2 > 0$ (derived from Thm.~\ref{thm : Matveev}).

As there are many $n\theta$ in $\mathcal{I}$ that are fairly
``evenly'' distributed, we are able to prove that \eqref{eq : example Baker end} cannot hold for all $n$. 
In other words, there must be some $n$ for which
\begin{equation*}
    0 < v_n < v_{n+4} < v_{n+160}
\end{equation*}
and $v_n < v_m < v_{n+160}$ implies that $m = n+4$.
Then, using the strategy of Lemma~\ref{lem : weak normality equivalent statement}, 
we translate this back to the sequence $\langle u_n\rangle_{n=0}^\infty$ to give the required result when $n$ is large enough.
In particular, we can calculate that when taking
\begin{equation*}
    n =218085867698737188268427463501308698889728969450963229999559\,,
\end{equation*}
 we have that $n\theta \in \mathcal{I}$, $n \equiv 3 \pmod 4$, and $u_n < u_m < u_{n+160}$ implies that $m = 4$.
 Thus, for some value $m \in \nat$, $\langle p_m, p_{m+1}, p_{m+2} \rangle = \langle u_n, u_{n+4}, u_{n+160} \rangle$ where all three terms are divisible by $5$.
 Thus, the pattern $\langle 0, 0, 0 \rangle$ does indeed appear in
 $\langle p_m\bmod 5 \rangle_{m=0}^\infty$. 

However, this construction is far from optimal. 
Although our choices of $b_2 = 4$ and $b_3 = 160$ were as small as possible with our choice of $\delta_2$ and $\delta_3$, $\delta_2$ and $\delta_3$ were not optimal choices. 
Still, that would give just a minor improvement as one has to be more careful with inequalities like~\eqref{eq : strong inequality sizes J_d}. 
Then one can find that choosing $b_2 = 8$, $b_3 = 28$, and $n = 16958443$, we indeed get that there is no $m \in \mathbb{N} \setminus\{0, 8, 28\}$ such that $u_m$ is between $u_n$, $u_{n+b_2}$ and $u_{n+b_3}$. 
Meanwhile, $\cos(n\theta) \approx 0.404$, $\cos((n+b_2)\theta) \approx 94.5$ and $\cos((n+b_3)\theta) \approx 751$, and so one should have had chosen much larger $\delta_2$ and $\delta_3$ for which the general bounds in Lem.~\ref{lem : upper bound length intervals} fail.

\subsection{Proof of the Main Theorems}\label{sec : Proof of weak normality}

In this section, we establish Thms.~\ref{thm : intro weak normality} and \ref{thm : MSO decidable}.
We first prove two technical lemmas. 
These two lemmas will both serve to prove these theorems but also to prove Lem.~\ref{lem : density weakly normal}, the continuous version of disjunctivity.

For the first lemma, we want that to show that each interval in $\mathbb{R}/(2\pi\mathbb{Z})$ contains intervals $\mathcal{J}_d(0, \sqrt{|\lambda|})$ and numbers $n\theta + \phi$ for relatively small numbers $d$ and $n$ in a certain congruent class.
The proof of this lemma in in App.~\ref{sec : omitted proofs}.

\begin{lemma}\label{lem : smallest n in interval}
    Let $T \ge 2$.
    There is a number $C_9 > 0$ such that for every $t \in \{0,\dots, T-1\}$ and small enough interval $\mathcal{I} \subset (-\pi/2, \pi/2) \subset \mathbb{R}/(2\pi\mathbb{Z})$, there are $|\mathcal{I}|^{-C_9} \le n_1, n_2 \le 2|\mathcal{I}|^{-C_9}$ such that $n_1 \theta + \phi \in \mathcal{I}$, $\mathcal{J}_{n_2}(0, |\sqrt{\lambda|}) \subset \mathcal{I}$ and $n_1, n_2 \equiv t\pmod T$.
\end{lemma}

For the fourth condition of Lem.~\ref{lem : density weakly normal}, we would like $D$ to be large. 
That is, the smallest $d$ such that $\mathcal{J}_d(1, \delta_\ell) \cap \mathcal{I}\ne\emptyset$ has to be quite large.
This is hard to guarantee, and our solution is to show that we can find a subinterval $\mathcal{I}' \subset \mathcal{I}$ for which this number $D$ is large. 
We show this in the following lemma whose proof is in App~\ref{sec : omitted proofs}.
\begin{lemma}\label{lem : better subintervals}
    Assume $\mathcal{I} \subset \mathbb{R}/(2\pi\mathbb{Z})$, $b_2,\dots,b_\ell \in \mathbb{N}$, $1< \delta_2 < \cdots < \delta_\ell <\sqrt{|\lambda|}$ and $D > 0$ satisfy the assertions of Lem.~\ref{lem : density weakly normal}.
    Then for every small enough $\varepsilon > 0$ there is a subinterval $\mathcal{I}'$ of $\mathcal{I}$ of length $\varepsilon$ for which the assertions of Lem.~\ref{lem : density weakly normal} hold for these $b_2,\dots,b_\ell$ and $\delta_1,\dots,\delta_\ell$ and some $D' > \varepsilon^{-1/2}$.
  \end{lemma}
  
With these last two lemmas, we have gathered all the tool needed to prove Lem.~\ref{lem : density weakly normal}, which is also done in App.~\ref{sec : omitted proofs}.
The idea of this proof is as described in Sec.~\ref{sec : example}:
we inductively construct the numbers $b_2,\dots,b_\ell$ and intervals 
\begin{equation*}
(-\pi/2,\pi/2) \supset \mathcal{I}_1 \supset \mathcal{I}_2 := \mathcal{I}_{b_2}(1,\delta_2) \supset \mathcal{I}_3 := \mathcal{I}_{b_3}(\delta_2,\delta_3) \supset \cdots \supset \mathcal{I}_{b_\ell}(\delta_{\ell-1},\delta_\ell)
\end{equation*}
such that on step $i$, we show that some $b_i$ exists for which the lemma holds for $b_2,\dots,b_i$ and $\mathcal{I}_i$.

Having established the continuous version of our result in Lem.~\ref{lem :
  density weakly normal}, we can solve the discrete version
Thm.~\ref{thm : intro weak normality} and conclude Thm.~\ref{thm : MSO
  decidable}: the MSO theory of the structure $\langle \nat;<,\{u_n
\colon n \in \nat\}\cap \nat \rangle$ is decidable.
The proof shares some of the main ideas of the proof of Lem.~\ref{lem : density weakly normal}.
\begin{proof}[Proof of Thms.~\ref{thm : MSO decidable} and~\ref{thm : intro weak normality}]
    As Thm.~\ref{thm : intro weak normality} implies Thm.~\ref{thm : MSO decidable} thanks to Lem.~\ref{lem : Theorems intro implications}, it is sufficient to prove Thm.~\ref{thm : intro weak normality}. 
    In turn, Lem.~\ref{lem : weak normality equivalent statement} states that Thm.~\ref{thm : intro weak normality} is implied by the  following statement: for all $N \in \nat$, $T \ge 2$ and $t_1,\dots,t_\ell \in \{0,\dots,T-1\}$, we can find $n_1,\dots,n_\ell \in \nat$ such that
    \begin{enumerate}
        \item $n_1,\dots,n_\ell \ge N$;
        \item $n_j \equiv t_j \pmod T$ for $1 \le j \le \ell$;
        \item $0 < v_{n_1} < \cdots < v_{n_\ell}$;
        \item for all $m \in \nat$ such that $v_{n_1} < v_m < v_{n_\ell}$, $m \in \{n_2,\dots,n_{\ell-1}\}$.
    \end{enumerate}
    We will prove this statement for given $N, T, \ell$, and $t_1,\dots,t_\ell$.

    Lem.~\ref{lem : density weakly normal} gives numbers $b_2,\dots,b_\ell \in \nat$, $1 < \delta_\ell < \sqrt{|\lambda|}$ and an interval $\mathcal{I}$.
    We take $n_1 = n$ and $n_j = n + b_j$ for $2 \le j \le \ell$. 
    
    When $n \ge N$, $n \equiv t_1 \pmod T$ and $n\theta + \phi \in \mathcal{I}$, imply items 1, 2, and 3.
    Indeed, items 2 and 3 hold for $n_1 = n$, and for $2 \le j \le \ell$, $n_j = n+b_j > n \ge N$ and
    \begin{equation*}
        n_j \equiv n_1 + b_j \equiv n_1 + (t_j - t_1) \equiv t_j \pmod T \, .
    \end{equation*}
    For item 3, we have that as $n\theta + \phi \in \mathcal{I}$,
    \begin{equation*}
        0 < \cos(n\theta + \phi) < \cos\big( (n+b_2)\theta+\phi\big)|\lambda|^{b_2} <\cdots < \cos\big((n+b_\ell)\theta +\phi\big)|\lambda|^{b_\ell}\,.
    \end{equation*}
    Multiplying these inequalities by $|\lambda|^n$ and noting that
    $\cos(m\theta + \phi)|\lambda|^m = v_m$ entails item~3.

    Let $0 < \varepsilon < |\mathcal{I}|$ be small enough.
    Lemma~\ref{lem : better subintervals} implies that there is an interval $\mathcal{I}_\varepsilon \subset \mathcal{I}$ of length $\varepsilon$ such that $\mathcal{I}_\varepsilon \cap \mathcal{J}_d(1, \delta_\ell) = \emptyset$ for all $d < \varepsilon^{-1/2}$.
    By Lem.~\ref{lem : smallest n in interval}, there is a $\varepsilon^{-C_9} < n < 2\varepsilon^{-C_9}$ such that $n\theta + \phi \in \mathcal{I}_\varepsilon$ and $n\equiv t_1 \pmod T$. 
    Thus, $n \ge N$ for small enough $\varepsilon$ and items 1, 2, and 3 are satisfied.

    Now assume that $m \not \in \{n_2,\dots,n_{\ell-1}\}$ and $v_{n_1} < v_m < v_{n_\ell}$. 
    Then, setting $d = m - n$, we have that $d \not\in \{b_2,\dots,b_{\ell-1}\}$ and 
    \begin{equation*}
        0 < \cos(n\theta + \phi)|\lambda|^n < \cos\big((n+d)\theta + \phi\big)|\lambda|^{n+d}  < \cos\big((n+b_\ell)\theta + \phi\big)|\lambda|^{n+b_\ell} \, .
    \end{equation*}
    As $n\theta + \phi \in \mathcal{I}_\varepsilon \subset \mathcal{I}$, $\cos((n+b_\ell)\theta + \phi)|\lambda|^{b_\ell} < \delta_\ell \cos(n\theta + \phi)$.
    Thus, $n\theta+\phi$ is in $\mathcal{J}_d(1, \delta_\ell)$ and so $d \ge \varepsilon^{-1/2}$ (which comes from Lem.~\ref{lem : better subintervals}).
    
    By Lem.~\ref{lem : upper bound length intervals}, $n\theta + \phi \in \mathcal{J}_d(1, \delta_\ell)$ and $-d\theta \pm \pi/2$ are at most $C_8|\lambda|^{-d}$ apart.
    Thus, for a constant $c_1$ derived from Thm.~\ref{thm : Matveev},
    \begin{align*}
        C_8|\lambda|^{-d} &\ge  |(n\theta + \phi) - (-d\theta \pm \pi/2)|_{\pi/2} \\
        &\ge |e^{i(n\theta + \phi) - (-d\theta \pm \pi/2)}| \\
        &> \frac{\pi}{2}|n + d|^{-c_1} \, .
    \end{align*}
    Taking logarithms,
    \begin{equation*}
        \log(C_8) + c_1\log(n + d) > d\log|\lambda| \, .
    \end{equation*}
    As $d \ge \varepsilon^{-1/2}$ and $n \ge \varepsilon^{-C_9}$, we have that $n, d \ge 2$ for small enough $\varepsilon$.
    Then, $dn\ge n+d$ and so 
    \begin{equation*}
        \log(C_8) + c_1\log(n)+c_1\log(d) > d\log|\lambda| \, .
    \end{equation*}
    Hence, either 
    \begin{equation*}
        2\log(C_8) +2c_1\log(d) > d\log|\lambda| \quad\text{or} \quad
        2c_1\log(n) > d\log|\lambda| \, .
    \end{equation*}
    The former is impossible for large enough $d$ (and thus small enough $\varepsilon$) while for the latter, the upper and lower bounds for $n$ and $d$ give that
    \begin{equation*}
         2c_1\log(2\varepsilon^{-C_9}) > 2c_1\log(n) > d\log|\lambda| > 
         \log|\lambda|\varepsilon^{-1/2} \, ,
    \end{equation*}
    which again is impossible for small enough $\varepsilon$.
    Hence item~4 also follows.\qedhere

  \end{proof}

  \section{Concluding Remarks}
  \label{sec : conclusion}

  Our main result, Thm.~\ref{thm : MSO decidable}, significantly
  expands the decidability landscape of MSO theories of structures of
  the form $\langle \nat ; <, P \rangle$, by considering predicates
  $P$ obtained as sets of positive terms of non-degenerate simple LRS
  having two dominant roots. A natural question is whether any of
  these constraints can be relaxed. It is conceivable that
  investigating simple LRS with three dominant roots might yield
  positive decidability results, although the present development
  would have to be significantly altered at various junctures.  We
  also note that open instances of the Skolem Problem~\cite{recseq}
  can easily be reduced from in the presence of simple LRS having four
  or more dominant roots. Likewise, considering non-simple LRS having
  three dominant roots or more exposes one to
  Positivity-hardness~\cite{OW14}. Let us remark that non-degeneracy
  is essential, as lifting this restriction has the effect of voiding
  the constraint on the number of dominant roots.\footnote{More
    precisely, one can construct a degenerate LRS for which the two
    dominant roots ``cancel'' each other out at every odd index, and the
    sub-LRS that remains over odd indices (involving the non-dominant
    roots) can be arbitrarily complex.} Finally, one might
  envisage expanding the present setup by adjoining further
  predicates. However, even in the simplest of cases, one expects
  quite formidable difficulties in attempting to follow that research
  direction; see~\cite{berthe2024decidability}.


\bibliography{main}

\begin{thebibliography}{10}

\bibitem{AghamovBKOP24}
Rajab Aghamov, Christel Baier, Toghrul Karimov, Jo{\"{e}}l Ouaknine, and Jakob Piribauer.
\newblock Linear dynamical systems with continuous weight functions.
\newblock In {\em {HSCC}}, pages 22:1--22:11. {ACM}, 2024.

\bibitem{Baier0JKLLOPW021}
Christel Baier, Florian Funke, Simon Jantsch, Toghrul Karimov, Engel Lefaucheux, Florian Luca, Jo{\"{e}}l Ouaknine, David Purser, Markus~A. Whiteland, and James Worrell.
\newblock The orbit problem for parametric linear dynamical systems.
\newblock In {\em {CONCUR}}, volume 203 of {\em LIPIcs}, pages 28:1--28:17. Schloss Dagstuhl - Leibniz-Zentrum f{\"{u}}r Informatik, 2021.

\bibitem{berthe2024decidability}
Val{\'{e}}rie Berth{\'{e}}, Toghrul Karimov, Joris Nieuwveld, Jo{\"{e}}l Ouaknine, Mihir Vahanwala, and James Worrell.
\newblock On the decidability of monadic second-order logic with arithmetic predicates.
\newblock In {\em {LICS}}, pages 11:1--11:14. {ACM}, 2024.

\bibitem{mso-notes}
Achim Blumensath.
\newblock Monadic {S}econd-{O}rder {M}odel {T}heory.
\newblock \url{http://www.fi.muni.cz/~blumens/}, 2024.
\newblock [Online; accessed on 24 January 2025].

\bibitem{emile1909probabilites}
{\'E}mile~M. Borel.
\newblock Les probabilit{\'e}s d{\'e}nombrables et leurs applications arithm{\'e}tiques.
\newblock {\em Rendiconti del Circolo Matematico di Palermo (1884-1940)}, 27(1):247--271, 1909.

\bibitem{braverman2006termination}
Mark Braverman.
\newblock Termination of integer linear programs.
\newblock In {\em {CAV}}, volume 4144 of {\em Lecture Notes in Computer Science}, pages 372--385. Springer, 2006.

\bibitem{buchi1966symposium}
J.~Richard B{\"u}chi.
\newblock Symposium on decision problems: On a decision method in restricted second order arithmetic.
\newblock In {\em Studies in Logic and the Foundations of Mathematics}, volume~44, pages 1--11. Elsevier, 1966.

\bibitem{Bug12}
Yann Bugeaud.
\newblock {\em Distribution modulo one and Diophantine approximation}, volume 193.
\newblock Cambridge University Press, 2012.

\bibitem{cartonthomas}
Olivier Carton and Wolfgang Thomas.
\newblock The monadic theory of morphic infinite words and generalizations.
\newblock {\em Information and Computation}, 176(1):51--65, 2002.

\bibitem{DCostaKMOSW22}
Julian D'Costa, Toghrul Karimov, Rupak Majumdar, Jo{\"{e}}l Ouaknine, Mahmoud Salamati, and James Worrell.
\newblock The pseudo-reachability problem for diagonalisable linear dynamical systems.
\newblock In {\em {MFCS}}, volume 241 of {\em LIPIcs}, pages 40:1--40:13. Schloss Dagstuhl - Leibniz-Zentrum f{\"{u}}r Informatik, 2022.

\bibitem{dubickas2016positive}
Art{\=u}ras Dubickas and Min Sha.
\newblock Positive density of integer polynomials with some prescribed properties.
\newblock {\em Journal of Number Theory}, 159:27--44, 2016.

\bibitem{elgot1966decidability}
Calvin~C. Elgot and Michael~O. Rabin.
\newblock Decidability and undecidability of extensions of second (first) order theory of (generalized) successor.
\newblock {\em The Journal of Symbolic Logic}, 31(2):169--181, 1966.

\bibitem{recseq}
Graham Everest, Alf van~der Poorten, Igor Shparlinski, and Thomas Ward.
\newblock {\em Recurrence Sequences}.
\newblock Mathematical Surveys and Monographs. American Mathematical Society, 2003.

\bibitem{Har02}
Glyn Harman.
\newblock One hundred years of normal numbers.
\newblock In {\em Surveys in Number Theory}, pages 57--74. AK Peters/CRC Press, 2002.

\bibitem{KarimovKNO023}
Toghrul Karimov, Edon Kelmendi, Joris Nieuwveld, Jo{\"{e}}l Ouaknine, and James Worrell.
\newblock The power of positivity.
\newblock In {\em {LICS}}, pages 1--11. {IEEE}, 2023.

\bibitem{KarimovKO022}
Toghrul Karimov, Edon Kelmendi, Jo{\"{e}}l Ouaknine, and James Worrell.
\newblock What's decidable about discrete linear dynamical systems?
\newblock In {\em Principles of Systems Design}, volume 13660 of {\em Lecture Notes in Computer Science}, pages 21--38. Springer, 2022.

\bibitem{KarimovLOPVWW22}
Toghrul Karimov, Engel Lefaucheux, Jo{\"{e}}l Ouaknine, David Purser, Anton Varonka, Markus~A. Whiteland, and James Worrell.
\newblock What's decidable about linear loops?
\newblock {\em Proc. {ACM} Program. Lang.}, 6({POPL}):1--25, 2022.

\bibitem{Kronecker}
L.~Kronecker.
\newblock Zwei {S}\"{a}tze \"uber {G}leichungen mit ganzzahligen {C}oefficienten.
\newblock {\em J. Reine Angew. Math.}, 53:173--175, 1857.

\bibitem{Kui12}
Lauwerens Kuipers and Harald Niederreiter.
\newblock {\em Uniform distribution of sequences}.
\newblock Courier Corporation, 2012.

\bibitem{LucaOW22}
Florian Luca, Jo{\"{e}}l Ouaknine, and James Worrell.
\newblock Algebraic model checking for discrete linear dynamical systems.
\newblock In {\em {FORMATS}}, volume 13465 of {\em Lecture Notes in Computer Science}, pages 3--15. Springer, 2022.

\bibitem{matveev2000explicit}
Eugene~M. Matveev.
\newblock An explicit lower bound for a homogeneous rational linear form in the logarithms of algebraic numbers. {II}.
\newblock {\em Izvestiya: Mathematics}, 64(6):1217, 2000.

\bibitem{OW14}
Jo{\"{e}}l Ouaknine and James Worrell.
\newblock Positivity problems for low-order linear recurrence sequences.
\newblock In {\em {SODA}}, pages 366--379. {SIAM}, 2014.

\bibitem{queffelec2006old}
Martine Queff\'elec.
\newblock Old and new results on normality.
\newblock {\em Lecture Notes-Monograph Series}, pages 225--236, 2006.

\bibitem{rabinovich}
Alexander Rabinovich.
\newblock On decidability of monadic logic of order over the naturals extended by monadic predicates.
\newblock {\em Information and Computation}, 205(6):870--889, 2007.

\bibitem{rabinovich2006decidable}
Alexander Rabinovich and Wolfgang Thomas.
\newblock Decidable theories of the ordering of natural numbers with unary predicates.
\newblock In {\em International Workshop on Computer Science Logic}, pages 562--574. Springer, 2006.

\bibitem{semenov1984logical}
Aleksei~Lvovich Semenov.
\newblock Logical theories of one-place functions on the set of natural numbers.
\newblock {\em Mathematics of the USSR-Izvestiya}, 22(3):587, 1984.

\bibitem{She75}
S.~Shelah.
\newblock The monadic theory of order.
\newblock {\em Ann. Math.}, 102:379--419, 1975.

\bibitem{Tho97}
Wolfgang Thomas.
\newblock Ehrenfeucht games, the composition method, and the monadic theory of ordinal words.
\newblock In {\em Structures in Logic and Computer Science}, volume 1261 of {\em Lecture Notes in Computer Science}, pages 118--143. Springer, 1997.

\bibitem{thomas1997languages}
Wolfgang Thomas.
\newblock Languages, automata, and logic.
\newblock In {\em Handbook of Formal Languages: Volume 3 Beyond Words}, pages 389--455. Springer, 1997.

\bibitem{shorey1984distance}
R.~Tijdeman, M.~Mignotte, and T.~N. Shorey.
\newblock The distance between terms of an algebraic recurrence sequence.
\newblock {\em J. Reine Angew. Math.}, 346:63--76, 1984.

\end{thebibliography}

\appendix

\section{Omitted proofs}\label{sec : omitted proofs}

\begin{proof}[Proof of Lem.~\ref{lem : dominant roots are nice}]
    As the characteristic polynomial of $\langle u_n\rangle_{n=0}^\infty$ has integer coefficients, $\overline{\lambda_1}$ and $\overline{\lambda_2}$ are also dominant characteristic roots. 
   Since $\langle u_n\rangle_{n=0}^\infty$ has exactly two dominant roots, $\{\lambda_1, \lambda_2\} = \{\overline{\lambda_1}, \overline{\lambda_2}\}$. 
    If $\lambda_1 = \overline{\lambda_1}$, $\lambda_2 = \overline{\lambda_2}$ and so both $\lambda_1$ and $\lambda_2$ are real. Hence, $\lambda_1/\lambda_2 = \pm1$, contradicting non-degeneracy. 
    Thus, $\lambda_2 = \overline{\lambda_1}$.

    If the argument of $\lambda_1$ is a rational multiple of $\pi$,
    $\lambda_1/\overline{\lambda_1}$ also has argument a rational
    multiple of $\pi$. Having modulus 1,
    $\lambda_1/\overline{\lambda_1}$ would then be a root of unity,
    contradicting non-degeneracy.

    Assume $|\lambda_1| \le 1$.
    By the Vieta formulas, the product of the absolute values of the
    characteristic roots is at most 1 and also equals $|c_d| \in
    \mathbb{Z}_{>0}$, where $c_d$ is the constant coefficient of the
    minimal polynomial of $\langle u_n\rangle_{n=0}^\infty$. 
    Hence, $|\lambda_1| = 1$ and there are no non-dominant characteristic roots. 
    Whence, by an old result of Kronecker~\cite{Kronecker}, we
    conclude that both $\lambda_1$ and $\lambda_2$ are roots of unity,
    and thus so is their quotient, contradicting once again non-degeneracy.
    Hence $|\lambda_1| > 1$ as claimed.
\end{proof}

\begin{proof}[Proof of Lem.~\ref{lem : First three conditions of LRS}]
    First, as discussed in Sec.~\ref{sec : LRS}, compute $r > 0$ and $0< R<|\lambda|$ such that $rR^n > r_n$ for all $n \in \nat$.

    To show that $\langle p_m\rangle_{m=0}^\infty$ is recursive, it is
    sufficient to find, for a given $k \in \nat$, a number $N$ such
    that $|u_n| > k$ for all $n \ge N$ as then $k \in P$ if and only if $k \in \{u_0,\dots,u_{N-1}\}$. 
    Using Thm.~\ref{thm : MST size LRS}, we have that when $n \ge C_2$ and $u_n = k$, 
    \begin{equation*}
        n\log|\lambda| - C_1 \log^2(n) \le \log|u_n| < \log(k + 1) \,  .
    \end{equation*}
    Hence $n$ is bounded and the desired $N$ can be obtained.

    To show that $P$ is infinite, we invoke Lem.~4
    from~\cite{braverman2006termination}: for infinitely many $n$,
    $\alpha \lambda^n + \overline{\alpha}\overline{\lambda}^n >
    c|\lambda|^n$ for some real $c > 0$.
    As $c|\lambda|^n > rR^n$ for all but finitely many $n$, there is
    $c' > 0$ such that $u_n \ge c' |\lambda|^n$ for infinitely many
    $n$.  Hence $P$ is indeed infinite. 
    
    It remains to show that $\langle p_m\rangle_{m=0}^\infty$ is effectively sparse.
    Assume $k, n_1, n_2\in \nat$, $n_1 > n_2$, and $|u_{n_1} - u_{n_2}| \le k$. 
    Then Thm.~\ref{thm : MST size difference terms of LRS} asserts that whenever $n_1 \ge C_4$,
    \begin{align*}
        \log(k + 1) &\ge \log|u_{n_1} - u_{n_2}|\\
        &\ge |\lambda|^{n_1} - C_3\log(n_1)^2\log(n_2+2) \\
        &\ge |\lambda|^{n_1} - C_3\log(n_1+1)^3 \, .
    \end{align*}
    Thus $|u_{n_1} - u_{n_2}| \le k$ implies that $n_1 \le N'$ for some computable constant $N'$. 
    Hence we can write out the set $P \cap \{0,\dots,k+1+ \max_{0 \le
      n \le N'} \{ u_n\}\}$ and find the two largest elements in this set having difference at most $k$. 
\end{proof}

\begin{proof}[Proof of Lem.~\ref{lem : Theorems intro implications}]
   By Lem.~\ref{lem : First three conditions of LRS}, $P$ is infinite,
   recursive, and effectively sparse, and so in order to apply Prop.~\ref{prop
     : reduction one unary predicate}, we only need to verify that
   $\Acc_{\langle p_m\bmod M\rangle_{m=0}^\infty}$ is decidable for all $M \ge 2$. 
    
    Let $M \ge 2$ and recall the definition \eqref{eq : definition S} of $S_M$.
    As stated in Sec.~\ref{sec : LRS}, $\langle u_n\bmod M\rangle_{n=0}^\infty$ is ultimately periodic modulo $M$, and both its period and preperiod can be effectively computed.
    Therefore, we can compute $S_M$ together with a number $N'$ such that for all $n \ge N'$, $u_n \equiv s \pmod M$ for some $s \in S_M$. 
    Next, compute $N_M$ large enough such that $p_{N_M} \ge u_n$ for all $0 \le n < N'$. 
    Then, by construction, for all $m \ge N_M$, $p_m \equiv s \pmod M$ for some $s \in S_M$. 

    Let $\mathcal{A}$ be a deterministic Muller automaton over the alphabet $\{0,\dots,M-1\}$. 
    After $\mathcal{A}$ has read $(p_0\bmod M),\dots,(p_{N_M-1}\bmod M)$, only elements from $S_M$ will be read.  
    Hence we can build a second deterministic Muller automaton $\mathcal{A}'$ over alphabet $S_M$ which accepts $\langle p_m\bmod M\rangle_{m=N_M}^\infty$ if and only if $\mathcal{A}$ accepts $\langle p_m\bmod M\rangle_{m=0}^\infty$ by hard-coding the initial segment and restricting the original alphabet $\{0,\dots,M-1\}$ to $S_M$.  
    By Thm.~\ref{thm : intro weak normality}, $\langle p_m\bmod M\rangle_{m=N_M}^\infty \in {S_M}^\omega$ is disjunctive, and thus by Thm.~\ref{thm : weakly normal decidable} we can determine whether $\mathcal{A}'$ accepts $\langle p_m\bmod M\rangle_{m=N_M}^\infty$. 
    Hence, combined with Lem.~\ref{lem : First three conditions of LRS}, the conditions of Prop.~\ref{prop : reduction one unary predicate} are met, yielding the desired result.
\end{proof}

\begin{proof}[Proof of Lem.~\ref{lem : upper bound length intervals}]
    Identify $\mathbb{R}/(2\pi\mathbb{Z})$ with $(-\pi,\pi]$.
    As $2\cos(x) = e^{ix} +e^{-ix}$ and $e^{id\theta}|\lambda|^d = \lambda^d$, for $\gamma \le \eta \le \delta$, we have that
    \begin{equation}\label{eq : point in interval}
      \cos(x + d\theta)|\lambda|^d = \eta \cos(x) \Longleftrightarrow
      \left(e^{ix}\right)^2 =
      -\frac{\overline{\lambda}^d-\eta}{\lambda^d-\eta} \, .
    \end{equation}
    Hence, there is a unique $x \in (-\pi/2, \pi/2]$ such that
    $\cos(x + d\theta)|\lambda|^d = \eta \cos(x)$.
    If $x = \pi/2$ and $\cos(x+d\theta) = \eta \cos(x)$, then $\cos(x)
    = 0$ and $\theta/\pi$ is rational, which we excluded in Lem.~\ref{lem : dominant roots are nice}. 
    We conclude that $\mathcal{J}_d(\gamma, \delta)$ is a single interval within $(-\pi/2, \pi/2)$.

    Let us now tackle the size of $\mathcal{J}_d(\gamma, \delta)$.
    As $\mathcal{J}_d(\gamma, \delta)$ consists of a single interval, $|\mathcal{J}_d(\gamma,\delta)| = |x_1 - x_2|_{2\pi}$, where $x_1$ and $x_2$ are solutions in $(-\pi/2,\pi/2)$ to \eqref{eq : point in interval} with $\eta = \gamma$ and $\eta = \delta$, respectively. 
    Using the triangle inequality on the unit circle, 
    \begin{equation*}
        |e^{ix_1} - e^{ix_2}| \le |x_1-x_2|_{2\pi} \le
        \frac{\pi}{2}|e^{ix_1} - e^{ix_2}| \, .
    \end{equation*}
    If $\gamma = \delta$, $x_1 = x_2$, and so by continuity, when $\delta - \gamma$ is small enough, $|e^{ix_1} - e^{ix_2}| < \sqrt{2}$.
    
    Assume we are on this boundary. 
    That is, $|e^{ix_1} - e^{ix_2}| = \sqrt{2}$.
    Then $e^{ix_1} =\pm i e^{ix_2}$.
    It follows that $e^{2ix_1} = -e^{2ix_2}$, and so
    \begin{align*}
      -1 &= e^{2ix_1}e^{-2ix_2} \\
       &= \frac{\overline{\lambda}^d-\gamma}{\lambda^d-\gamma}\cdot \frac{\lambda^d-\delta}{\overline{\lambda}^d-\delta}\\
        &=\frac{\lambda^d\overline{\lambda}^d -\gamma \lambda^d -\delta \overline{\lambda}^d +\gamma\delta}{\lambda^d\overline{\lambda}^d - \delta \lambda^d -\gamma \overline{\lambda}^d +\gamma\delta} \, .
    \end{align*}
    Therefore, $2\lambda^d\overline{\lambda}^d -(\delta-\gamma) \lambda^d
       +(\delta-\gamma) \overline{\lambda}^d +2\gamma\delta = 0. $
    Then $|\lambda|^{2d} \le (\delta - \gamma)|\lambda|^d + \gamma\delta$, and so $(|\lambda|^d + \gamma)(|\lambda|^d - \delta) \le 0$. 
    This is impossible in our scenario as $0 < \gamma < \delta < \sqrt{|\lambda|}$ and $d \ge 1$.
    Thus, again by continuity in $\gamma$ and $\delta$, $|e^{ix_1} - e^{ix_2}| < \sqrt{2}$.
    
    By the geometry of the unit circle, $|e^{ix_1} - e^{ix_2}| < \sqrt{2}$ implies that $\sqrt{2} \le |e^{ix_1} + e^{ix_2}| \le 2$. 
    Then, using the fact that $|e^{2ix_1}-e^{2ix_2}| = |e^{ix_1}-e^{ix_2}|\cdot|e^{ix_1}+e^{ix_2}|$, we obtain 
    \begin{equation*}
         \frac{1}{2}\big|e^{i2x_1} - e^{i2x_2}\big| \le |x_1-x_2|_{2\pi} 
        \le \frac{\pi}{2\sqrt{2}} \big|e^{i2x_1} - e^{i2x_2}\big| \, .
    \end{equation*}
    We can rewrite $|e^{i2x_1} - e^{i2x_2}|$ as follows:
    \begin{align*}
        \big|e^{2ix_1}-e^{2ix_2}\big| &= \left|\frac{\overline{\lambda}^d-\gamma}{\lambda^d-\gamma}-\frac{\overline{\lambda}^d-\delta}{\lambda^d-\delta}\right|\\
        &= \frac{\left|\big(\overline{\lambda}^d-\gamma\big)\big(\lambda^d-\delta)-\big(\overline{\lambda}^d-\delta\big)\big(\lambda^d-\gamma)\right|}{|\lambda^d-\gamma||\lambda^d-\delta|} \\
        &= \frac{1}{|\lambda^d-\gamma||\lambda^d-\delta|}\left|(\delta-\gamma)\lambda^d-(\delta-\gamma)\overline{\lambda}^d\right|\\
        &= \frac{(\delta-\gamma)|\lambda|^d}{|\lambda^d-\gamma||\lambda^d-\delta|}\left|e^{i2d\theta} - 1\right| \, .
    \end{align*}
    As $\gamma, \delta < \sqrt{|\lambda|}$, there are constants $c_1, c_2 > 0$ such that $c_1 |\lambda|^d < |\lambda^d - \eta| < c_2 |\lambda|^d$ for $\eta \in \{\gamma, \delta\}$ and $d \ge 1$.
    Thus
    \begin{equation}\label{eq : upper estimate for J}
        |\mathcal{J}_d(\gamma,\delta)| \le
        \frac{\pi}{2\sqrt{2}}\big|e^{2ix_1}-e^{2ix_2}\big| <
        C_5\frac{\delta-\gamma}{|\lambda|^d}
    \end{equation}
    for $C_5 = \frac{\pi}{2\sqrt{2}c_1^2}$.
    Similarly, 
    \begin{equation*}
    |\mathcal{J}_d(\gamma,\delta)| \ge \frac{1}{2}\big|e^{ix_1}-e^{ix_2}\big|  =
    \frac{(\delta-\gamma)|\lambda|^d}{2|\lambda^d-\gamma||\lambda^d-\delta|}\big|e^{i2d\theta} - 1\big|\,.
    \end{equation*}
    Thus,
    \begin{equation*}
    |\mathcal{J}_d(\gamma, \delta)| >  C_6 \frac{\big|e^{i2d\theta} - 1\big|(\delta-\gamma)}{|\lambda|^d}
    \end{equation*}
    for a constant $C_6$.
    Applying Thm.~\ref{thm : Matveev} on the latter, we obtain a constant $C_7$ such that $|e^{i2d\theta} - 1| > d^{-C_7}$, and the result follows.
    
    For the last claim, we estimate $|\mathcal{J}_d(0,\sqrt{|\lambda|})|$ with \eqref{eq : upper estimate for J}. 
\end{proof}

\begin{proof}[Proof of Lem.~\ref{lem : smallest n in interval}]
    By the pigeonhole principle, there are distinct $d_1, d_2 \in \nat$ such that $0 \le d_1, d_2 \le \lceil \pi|\mathcal{I}|^{-1} \rceil$ and $d_1 T\theta$ and $d_2T\theta$ have distance at most $|\mathcal{I}|$ modulo $2\pi$. 
    Here, we use that $\theta$ is not a rational multiple of $\pi$ (Lem.~\ref{lem : dominant roots are nice}).
    The last condition implies that
    \begin{equation*}
         \big| (d_1-d_2)T\theta\big|_{2\pi} < |\mathcal{I}|
    \end{equation*}
    Baker's theorem (Thm.~\ref{thm : Matveev}) implies there is a computable number $c_1$ such that
    \begin{align*}
          \big| (d_1-d_2)T\theta\big|_{2\pi} 
          &> \big|e^{iT(d_1-d_2)\theta} - 1\big| \\
          &> |d_1 - d_2|^{-c_1} \\
          &> \big\lceil\pi|\mathcal{I}|^{-1}\big\rceil^{-c_1} \, .
    \end{align*}
    For all $N \in \mathbb{Z}$ and $x \in \mathbb{R} /(2\pi\mathbb{Z})$, there is a $N \le n \le 2\pi \lceil \pi|\mathcal{I}|^{-1} \rceil^{c_1} + N$ such that $x + nT(d_1-d_2) \theta \in \mathcal{I}$. 
    Taking $C_9$ slightly bigger than $c_1$, $ 2\pi T \lceil \pi|\mathcal{I}|^{-1} \rceil^{c_1} < |\mathcal{I}|^{-C_9}$ for small enough $\mathcal{I}$. 
    Hence, for all $N \in \mathbb{Z}$ and $x \in \mathbb{R} /(2\pi\mathbb{Z})$, there is a $n \equiv t \pmod T$ such that $x + n\theta\in\mathcal{I}$ and $N \le n \le |\mathcal{I}|^{-C_9} + N$. 
    For $n_1$, let $x = \phi$ and $N = |\mathcal{I}|^{-C_9}$.

    Let $\mathcal{I}'$ be the middle half of $\mathcal{I}$. 
    As before, there is an $n_2 \equiv t \pmod{T}$ such that
    $|\mathcal{I}'|^{-C_9} \le n_2 \le 2|\mathcal{I}'|^{-C_9}$ and $-n_2\theta \pm \pi/2 \in \mathcal{I}'$. 
    If $\mathcal{J}_{n_2}(0, \sqrt{|\lambda|}) \not\subseteq\mathcal{I}$, then as $\mathcal{J}_{n_2}(0, \sqrt{|\lambda|})$ has a point in $\mathcal{I}$ and another one outside of $\mathcal{I}$.
    So, $|\mathcal{J}_{n_2}(0, \sqrt{|\lambda|})| \ge \frac{1}{4}|\mathcal{I}|$ and 
    Lem.~\ref{lem : upper bound length intervals} implies that
    \begin{align*}
        \frac{1}{4}|\mathcal{I}| \le |\mathcal{J}_{n_2}(0, \sqrt{|\lambda|})| &\le C_8|\lambda|^{-n_2} \\
        &\le C_8|\lambda|^{-|\mathcal{I}'|^{-C_9}} 
    \end{align*}
    After taking logarithms, we have that
    \begin{equation*}
        \log|\mathcal{I}| \le \log(4C_8) - (|\mathcal{I}'|)^{-C_9}\log|\lambda|= \log(4C_8) - (|\mathcal{I}|/2)^{-C_9}\log|\lambda|\,,
    \end{equation*}
    which cannot hold when $\mathcal{I}$ is small enough.
    Hence $\mathcal{J}_{n_2}(1, \sqrt{|\lambda|}) \subseteq \mathcal{I}$.
    The lemma follows.
\end{proof}

\begin{proof}[Proof of Lem.~\ref{lem : better subintervals}]
      For an interval $\mathcal{I}'\subset (-\pi/2,\pi/2)$, let $D(\mathcal{I}')$ denote the smallest natural number $d \ge 1$ such that $d \ne b_2,\dots,b_\ell$ and $\mathcal{J}_d(1, \delta_\ell) \cap \mathcal{I}'$ is non-empty.
      Set 
      \begin{equation*}
          c_1 = |\mathcal{I}| - \sum_{\substack{d = D \\ d \not\in \{b_2, \dots, b_\ell \}}}^\infty|\mathcal{J}_d(1, \delta_\ell)| \, .
      \end{equation*}
      By construction, $c_1 > 0$.
      Let $k \in \mathbb{N}$ such that $\frac{c_1}{k+2} < \varepsilon \le \frac{c_1}{k+1}$ (which always exists for small enough $\varepsilon$).
      We will study the set 
      \begin{equation*}
          X = \mathcal{I}\setminus \bigcup_{\substack{d = D \\ d \not\in \{b_2, \dots, b_\ell \}}}^k \mathcal{J}_d(1, \delta_\ell)\,,
      \end{equation*}
      which is an interval $\mathcal{I}$ from which at most $k$ intervals are removed. 
      Thus, $X$ consists out of at most $k+1$ intervals.
      By assumption, $|X| > c_1$, and so $|X|$ contains an interval $\mathcal{I}'$ of length $\varepsilon \le \frac{c_1}{k+1}$ such that $D(\mathcal{I}') \ge k+1$.
      Then, for large enough $k$, we have that
      \begin{equation*}
          D(I') \ge k+1 > \sqrt{\frac{k+2}{c_1}} > \varepsilon^{-1/2}\,.
      \end{equation*}
      When $\varepsilon$ is small enough, the result follows when setting $D' = D(\mathcal{I}')$.
  \end{proof}

\begin{proof}[Proof of Lem.~\ref{lem : density weakly normal}]
    We apply induction on $k$ and require that the conditions hold for $1,\dots,k-1$ for an interval $\mathcal{I}_{k-1}$, and construct an interval $\mathcal{I}_k$ that satisfies the theorem when items~(a), (b), and~(c) are restricted to $1,\dots,k$. 
    Then we take $\mathcal{I} = \mathcal{I}_\ell$.

    For the base case, let $\mathcal{I}_1 = \{x \in \mathbb{R}/(2\pi\mathbb{Z}) \colon \cos(x) > |\lambda|^{-1}\}$, $D = 1$, and $1 < \delta_\ell < \sqrt{|\lambda|}$ small enough such that that the last inequality of  
    \begin{equation*}
        \sum_{d = D}^\infty |\mathcal{J}_d(1, \delta_\ell)| \le
        \sum_{d=1}^\infty C_5 \frac{\delta_\ell-1}{|\lambda|^d} \le
        C_5 \frac{\delta_\ell-1}{|\lambda|-1}  < |\mathcal{I}_1| 
    \end{equation*}
    holds, where the first inequality follows from Lem.~\ref{lem : upper bound length intervals} and the second by a geometric series.
    Hence item~(d) holds, and items~(a), (b), and~(c) follow by construction.
    The base case follows.

    For the other cases, choose $\delta_2,\dots,\delta_{\ell-1}$ such that $1 < \delta_2 < \cdots < \delta_{\ell-1} < \delta_\ell$. Furthermore, for simplicity, set $\delta_1 = 1$.
    
    Now let $k \ge 1$ and $\varepsilon > 0$.
    For $\varepsilon$ small enough, applying Lem.~\ref{lem : better subintervals} on $\mathcal{I}_{k-1}$ gives intervals $\mathcal{I}_\varepsilon \subset \mathcal{I}_{k-1}$ of length $\varepsilon$ where the smallest $d$ such that $\mathcal{J}_d(1,\delta_\ell)\cap\mathcal{I}_\varepsilon \ne \emptyset$ is at least $\varepsilon^{-1/2}$.
    
    Lemma~\ref{lem : smallest n in interval} implies that when $\varepsilon$ is again small enough, there is a $\varepsilon^{-C_9} < b_k < 2\varepsilon^{-C_9}$ such that $b_k \equiv t_k - t_1 \pmod T$ and $\mathcal{J}_{b_k}(\delta_{k-1},\delta_k) \subset \mathcal{I}_\varepsilon$. 

    If  $d \ne 0, b_2,\dots, b_k$ and $\mathcal{J}_{b_k}(\delta_{k-1}, \delta_k) \cap \mathcal{J}_d(1, \delta_\ell) \ne \emptyset$, then $d > \varepsilon^{-1/2}$.
    By Lem.~\ref{lem : upper bound length intervals}, 
    \begin{equation*}
        |(b_k - d)\theta \pm_1 \pi/2 \pm_2 \pi/2|_{2\pi} \le
        \frac{C_8}{|\lambda|^{b_k}} + \frac{C_8}{|\lambda|^d}\le
        \frac{2C_8}{|\lambda|^{\min\{b_k, d\}}} \, .
    \end{equation*}
    We put $j \pi = \pm_1 \pi/2 \pm_2 \pi /2$ for $j \in
    \mathbb{Z}$.\footnote{We have adorned the $\pm$ operator with
      subscripts ($1$ and $2$) to indicate how the particular choice of
      signs should be preserved.}
    Then by Baker's theorem (Thm.~\ref{thm : Matveev}), there is a constant $c_1 > 0$,
    \begin{align*}
        |b_k-d|^{-c_1} 
        &< |e^{i( (b_k-d)\theta+j \pi)} - 1| \\
        &\le |(b_k - d)\theta+j \pi|_{2\pi} \\
        &\le \frac{2C_8}{|\lambda|^{\min\{b_k, d\}}} \, .
    \end{align*}
    Taking logarithms, we have:
    \begin{equation}\label{eq : random eq in weak normality lemma}
         \min\{b_k, d\} \log|\lambda| < \log(2C_8) + c_1\log|b_k-d| \,
         .
    \end{equation}
    Hence, if $d < b_k$, $|b_k-d| < b_k$ (as $d > 0$) and
    \begin{equation*}
        d \log|\lambda| < \log(2C_8) + c_1 \log(b_k) \, .
    \end{equation*}
    Using that $d > \varepsilon^{-1/2}$ and $b_k \le 2(\varepsilon/2)^{-C_9}$, we obtain
    \begin{equation*}
         \varepsilon^{-1/2} \log|\lambda| < \log(2C_8) + c_1\log(2)
         - c_1 C_9 \log(\varepsilon/2) \, .
    \end{equation*}
    This is impossible for sufficiently small $\varepsilon$.
    We can therefore assume that $d > b_k$.
    We take $\mathcal{I}_k := \mathcal{J}_{b_k}(\delta_{k-1},\delta_k)$ such that items~(a) and~(b) are satisfied. 
    Now assume $d > b_k$ and let $D = d$ (and so item~(c) automatically holds).
    For a contradiction, assume item~(d) is violated.
    Lemma~\ref{lem : upper bound length intervals} and the geometric series give that for some $c_2 >0$,
    \begin{equation*}
         \frac{C_6(\delta_k - \delta_{k-1})}{|\lambda|^{b_k}b_k^{C_7}}
         \le |\mathcal{I}_k| \le \sum_{d' = d}^\infty
         |\mathcal{J}_{d'}(1,\delta_\ell)| < c_2|\lambda|^{-d} \, .
    \end{equation*}
    Hence, setting $c_3 = \log(\frac{c_2}{C_6(\delta_k-\delta_{k-1})})$ and taking logarithms,
    \begin{equation}\label{eq : second to last induction}
         (d - b_k)\log|\lambda| \le c_3 + C_7\log(b_k) \, .
    \end{equation}
    Inserting the latter in \eqref{eq : random eq in weak normality lemma} gives
    \begin{equation*}
         b_k\log|\lambda| \le \log(2C_8) + c_1\log\left(\frac{c_3 +
             C_7\log b_k}{\log|\lambda|}\right) \, ,
    \end{equation*}
    which upper bounds $b_k$ (independently of $\varepsilon$).
    As taking $\varepsilon$ small gives arbitrarily large $b_k$, the result follows for $k$. 
    Induction completes the proof.
\end{proof}

\end{document}